\documentclass[journal]{IEEEtran}
\pdfoutput=1

\usepackage{fancyhdr}
\usepackage{hyperref}
\usepackage{amsmath, amsthm, amssymb, graphicx}
\usepackage{xspace}
\usepackage{braket}
\usepackage[font={small}]{caption}
\usepackage{color}
\usepackage{setspace}
\usepackage{enumerate}
\usepackage{graphics}
\usepackage{cite}
\usepackage{dsfont}
\usepackage{latexsym}
\usepackage{amsmath}
\usepackage{amsfonts}
\usepackage{amssymb}
\usepackage{times}
\usepackage{sgame}
\usepackage{color}
\usepackage{verbatim}
\usepackage{xcolor}
\usepackage{pgfplots}
\usepackage{epsfig} 
\usepackage{subcaption}

\pgfplotsset{width=7cm,compat=newest}
\usepgfplotslibrary{fillbetween}
\usepackage{tikz}
\usetikzlibrary{positioning,backgrounds}
\usetikzlibrary{shapes,snakes}
\usepackage{hyperref}

\usepackage{enumitem}
\newlist{inparaenum}{enumerate}{2}% allow two levels of nesting in an enumerate-like environment
\setlist[inparaenum]{nosep}% compact spacing for all nesting levels
\setlist[inparaenum,1]{label=\bfseries\arabic*.}% labels for top level
\setlist[inparaenum,2]{label=\arabic{inparaenumi}\emph{\alph*})}% labels for second level
\newcommand{\be}{\begin{equation}}
\newcommand{\ee}{\end{equation}}
\newcommand{\mbb}[1]{\mathbb{#1}}

\newcommand{\mcal}[1]{\mathcal{#1}}

\newtheorem{theorem}{Theorem}[section]

\newtheorem{lemma}{Lemma}[section]

\newtheorem{fact}{Fact}[section]

\newtheorem{definition}{Definition}
\newtheorem{result}{Result}

\def\bs{\boldsymbol}

\def\sf{S}

\normalsize\title{\LARGE \bf
	Strategically revealing intentions in General Lotto games
	\thanks{This work is supported by UCOP Grant LFR-18-548175, ONR grant \#N00014-20-1-2359, AFOSR grants \#FA9550-20-1-0054 and \#FA9550-21-1-0203, and the Army Research Lab through the ARL DCIST CRA \#W911NF-17-2-0181. The material in this paper contains and extends a subset of results from \cite{Chandan_2020}.}}

\author{
	Keith Paarporn, Rahul Chandan, Dan Kovenock, Mahnoosh Alizadeh, and Jason R. Marden
	\thanks{K. Paarporn, R. Chandan, M. Alizadeh, and J.R. Marden are with the Department of Electrical and Computer Engineering, University of California, Santa Barbara. D. Kovenock is with the Economic Science Institute, Argyros School of Business and Economics at Chapman University. Contact: \texttt{ \{kpaarporn,rchandan,alizadeh,jrmarden\}@ucsb.edu, kovenock@chapman.edu}.
	}
}

\begin{document}
\maketitle

\begin{abstract}
	Strategic decision-making in uncertain and adversarial environments is crucial for the security of modern systems and infrastructures.  A salient feature of many optimal decision-making policies is a level of unpredictability, or randomness, which helps to keep an adversary uncertain about the system's behavior.  This paper seeks to explore decision-making policies on the other end of the spectrum -- namely, whether there are benefits in revealing one's strategic intentions to an opponent before engaging in  competition. We study these scenarios in a well-studied model of competitive resource allocation problem known as General Lotto games. In the classic formulation, two competing players simultaneously allocate their assets to a set of battlefields, and the resulting payoffs are derived in a zero-sum fashion. Here, we consider a multi-step extension where one of the players has the option to publicly pre-commit assets in a binding fashion to battlefields before play begins. In response, the opponent decides which of these battlefields to secure (or abandon) by matching the pre-commitment with its own assets. They then engage in a General Lotto game over the remaining set of battlefields. Interestingly, this paper highlights many scenarios where strategically revealing intentions can actually significantly improve one's payoff.  This runs contrary to the conventional wisdom that randomness should be a central component of decision-making in adversarial environments.
\end{abstract}

%%%%%%%%%%%%%%%%%%%%%%%%%%%%%%%%%%%%%%%%%%%%%%%%%%%%%%%
\section{Introduction}
% !TEX root = main.tex

Society is increasingly reliant on autonomous technologies ingrained in critical infrastructures and socio-technical systems, made possible by advances in computing, communication, and control. Ensuring the security of these systems against adversaries poses several problems that must be addressed. As such, strategic decision-making in adversarial environments is consequential to many problems, including the security of cyber-physical systems, perimeter defense, and surveillance \cite{Abdallah_2021,Wu_2021_TAC,VonMoll_2020,Vamvoudakis_2014}. One well-studied framework for studying such systems is zero-sum games where the primary focus centers on characterizing optimal (max-min) strategies for competing players \cite{vonNeumann,alpcan2010network,Vamvoudakis_2017,Basar_1976}. Several different formulations of zero-sum games have been investigated in the control-theoretic literature encompassing asymmetric information, dynamics, and team-based decision-making \cite{Li_2020,Kartik_2021asymmetric,Kartik_2021teams}. In many of these formulations, classic methodologies in feedback and control have been instrumental in the derivation of optimal strategies.

% For example, random scheduling of security patrols are employed for the protection of airports, national borders, and wildlife \cite{tambe2011security,pita2008deployed,yang2014adaptive}. 

This paper focuses on a class of zero-sum games (equivalently, constant-sum games) that models resource allocation in competitive scenarios.  In particular, two players must strategically allocate their limited resources over a number of battlefields. A player wins the battlefields (and their associated values) that it sends more resources to, and the objective is to maximize the total value of secured battlefields. These scenarios are known in the literature as "Colonel Blotto" games. The equilibrium strategies are notoriously difficult to characterize, and general solutions to the Blotto game are still an open problem \cite{Thomas_2018}. For this reason, researchers often study the General Lotto game, which only requires players to spend their resource budgets in expectation \cite{Bell_1980,Myerson_1993,Hart_2008,Kovenock_2020}. It is more amenable to analysis while maintaining essential aspects of competitive resource allocation. The Colonel Blotto game and its many variants can naturally be applied to cyber-physical and networked system security \cite{Shahrivar_2014,Ferdowsi_2020,Guan_2019}, economic competition, and military strategy \cite{Borel,Gross_1950,Blackett_1958,Bell_1980,Shubik_1981,Roberson_2006,Golman_2009}.

% The Colonel Blotto game is a well-known model of competitive resource allocation that is often formulated as a zero-sum or constant-sum game.   In most cases of interest (i.e. when players' resource budgets are not too asymmetric), pure strategy equilibria do not exist, and thus optimal mixed strategies must be characterized. As such, deterministic (or pure) strategies are exploitable by an opponent. 

\begin{figure*}[t]
	\centering
	
	\begin{subfigure}{.3\textwidth}
		\centering
		\includegraphics[scale=0.16]{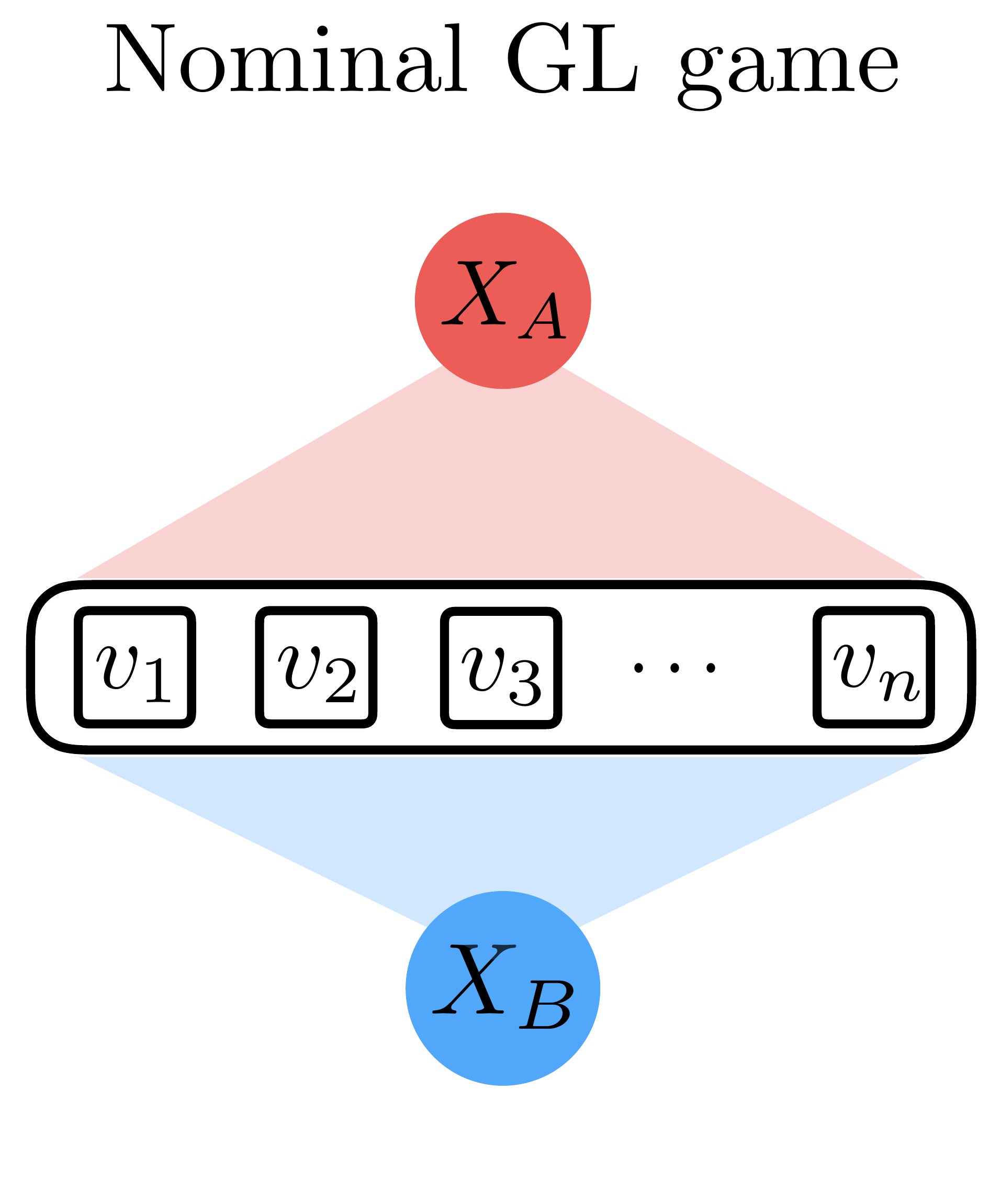}
		\caption{}\label{fig:GL_diagram}
	\end{subfigure}  
	\begin{subfigure}{.68\textwidth}
		\centering
		\includegraphics[scale=0.16]{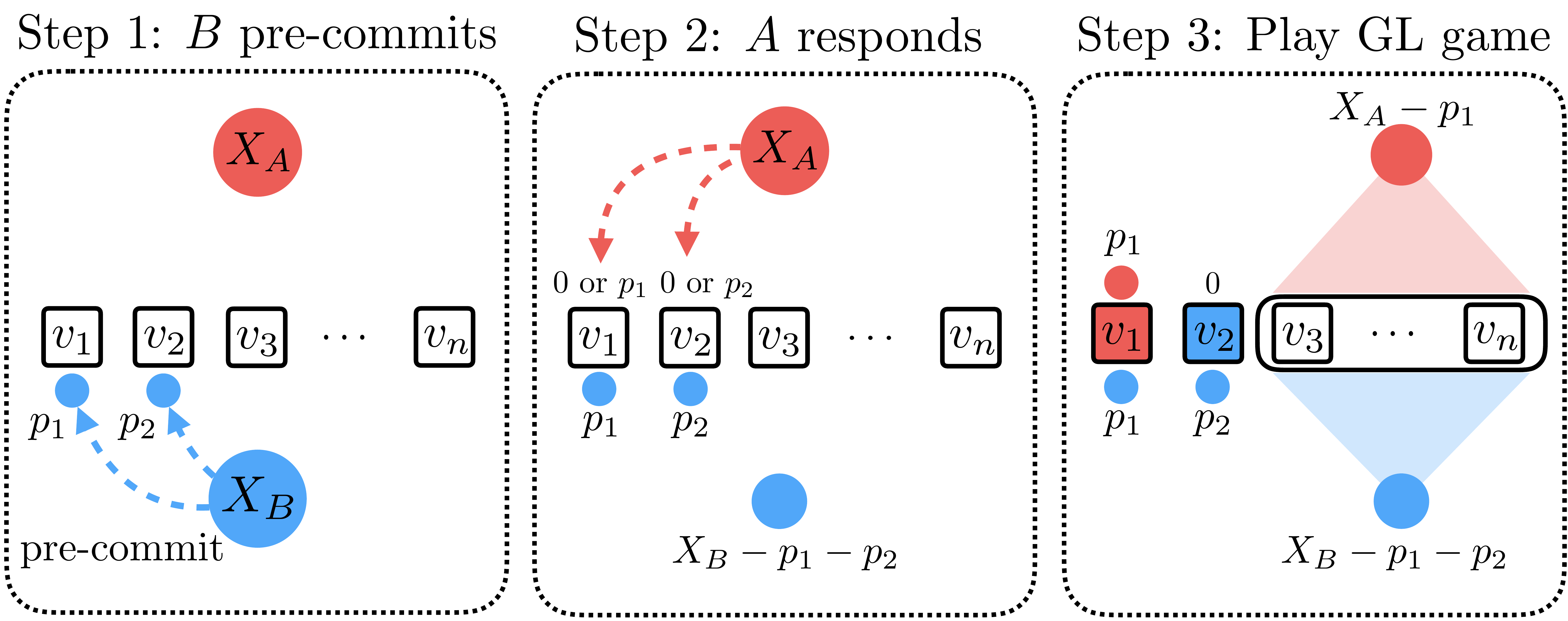}
		\caption{}\label{fig:precommit_diagram}
	\end{subfigure}
	\caption{Model of public pre-commitments. (a) For simplicity, we depict General Lotto games where battlefield valuations are symmetric across players in these diagrams, i.e. $\text{GL}(X_A,X_B,\bs{v})$. This is the nominal interaction if no player has the option to pre-commit resources, or the pre-committing player chooses not to pre-commit any resources. (b) Diagrams showing the sequence of events when player $B$ has the option to pre-commit. In step 1, $B$ pre-commits $p_1$ and $p_2$ resources to battlefields 1 and 2, respectively. In step 2, $A$ publicly responds to the pre-commitments by either matching or withdrawing. It secures the battlefields it matches, and loses the ones it withdraws from. In step 3, the players engage in a GL game on the remaining set of battlefields with their remaining resources.}
	\label{fig:diagrams}
\end{figure*}

A salient feature of optimal decision-making  policies in Colonel Blotto games is  a  level  of  unpredictability,  or  randomness, which  keeps  an  opponent  uncertain  about  one's behavior.  In this paper, we consider whether it is ever beneficial to reveal one's strategic intentions to an opponent. At first thought, doing so may only hurt one's position in a competition as randomization makes one's strategy less exploitable by an opponent. Indeed, in the Colonel Blotto game, pure strategy equilibria do not exist in most interesting cases (i.e. when players' resource budgets are not too asymmetric), and thus optimal mixed strategies must be characterized. However, there are practical contexts where revealing one's strategies do provide benefits. For example, shows of force are commonly used in military operations to discourage an enemy from engaging in conflict.

To study this aspect of adversarial interaction, we consider a sequenced formulation of General Lotto games, in which one player has the option to pre-commit resources to battlefields. These battlefields are then subject to a take-it or leave-it rule -- the opponent, in response, must decide which battlefields to secure by matching the pre-commitment with its own resources, and which battlefields to withdraw from entirely. After the opponent makes a decision, both players subsequently engage in a conventional General Lotto game on the remaining set of battlefields.

We are primarily concerned about whether there are benefits for a single competitor to publicly pre-commit resources in the multi-stage formulation.   We investigate pre-commitments under two different contexts. 1) The standard General Lotto game, which is a one-vs-one scenario over a set of $n$ battlefields that are commonly valued by both players.  2) A generalized General Lotto game, which allows for battlefield valuations to be asymmetric across players \cite{Kovenock_2020}. We find there are benefits to pre-commit in both contexts, though the types of benefits and conditions for them to exist differ. Most notably, a weaker-resource player never has an incentive to pre-commit in the first scenario. A weaker-resource player can have incentives to pre-commit in the second scenario. Before providing a summary of our results and contributions, we first introduce General Lotto games.

% OLD
% We find there are incentives to pre-commit in all three contexts, though the conditions for incentives to exist are very different in each one. A summary of our results and contributions are given next.

\subsection{Background on General Lotto games}\label{sec:GL_background}

A General Lotto game consists of two players $i=A,B$, each with limited resource endowments $X_A,X_B$. They compete over a set of battlefields $\mcal{B} = \{1,\ldots,n\}$. The endowments are use-it-or-lose-it, as there is no opportunity cost for the players to use resources outside of the game. An allocation for player $i$ is a vector $\bs{x}_i = \{x_{i,b}\}_{b\in\mcal{B}} \in \mbb{R}_+^n$, which describes the division of $i$'s resources to the $n$ battlefields. Player $i$ holds a valuation $v_{i,b} \geq 0$ for each battlefield $b \in \mcal{B}$. Let us denote $\bs{v}_i = \{v_{i,b}\}_{b\in\mcal{B}}$ as the vector of player $i$'s valuations. Player $i$ wins battlefield $b$ and the value $v_{i,b}$ if it allocates more resources than its opponent. The losing player on $b$ gets zero value. An admissible strategy for player $i$ is any $n$-variate distribution $F_i$ over $\mbb{R}_+^n$ (non-negative real vectors) that satisfies the condition  
\begin{equation}\label{eq:lotto_constraint}
\mathbb{E}_{\bs{x}_i \sim F_{i}}\left[ \sum_{b\in\mcal{B}} x_{i,b} \right] \leq X_i.
\end{equation}
In words, a player can randomize over any set of allocations, as long as the resources expended do not exceed its endowment \emph{in expectation}. We refer to an instance of the game with $\text{GL}(X_A,X_B,\bs{v}_A,\bs{v}_B)$.

We say the battlefield valuations are \emph{symmetric} across players if $v_{A,b} = v_{B,b}$ for all $b \in \mcal{B}$, and \emph{asymmetric} otherwise\footnote{The equilibrium strategies in the symmetric case are also invariant to unilateral scaling of one of the player's valuations, i.e. if it holds that $v_{A,b} = K\cdot v_{B,b}$ $\forall b\in\mcal{B}$ and for some $K > 0$. Here,  the players' valuations of the battlefields are relatively identical. The analysis of \cite{Kovenock_2020} is required for cases of asymmetric and relatively different valuations.}. In the symmetric case, there is an equilibrium payoff for both players. We refer to an instance of the General Lotto game in the symmetric case with $\text{GL}(X_A,X_B,\bs{v})$, where the players' valuations are encoded in a single valuation vector $\bs{v}$ (see Figure \ref{fig:GL_scenario}). The equilibrium strategies and payoffs for all instances of General Lotto games have been characterized in the literature \cite{Hart_2008,Kovenock_2012,Kovenock_2020}. Multiple payoff-distinct equilibria can arise when the relative valuations of battlefields are asymmetric across players \cite{Kovenock_2020}. 

% The GL game is a well-known variant of the Colonel Blotto game, where mixed strategies are $n$-variate distributions with support confined to the set of \emph{feasible} allocations $\Delta_n(X_i) = \{\bs{x} \in \mbb{R}_+^n: \sum_{b=1}^n x_b \leq X_i\}$. Characterizing equilibria for Colonel Blotto games in general is known to be a challenging problem, largely due to the support constraint for mixed strategies.

%%%%%%%%%%%%%%%%%%%%%%%%%%%%%%%%%%%%%%%%%%%%%%%%%%%%%%%%%%%%%%%%%%%%%%%%%%%%%%%%

% \input{lotto_diagrams}

\begin{figure*}[t]
	\centering
	\begin{subfigure}{.32\textwidth}
		\centering
		\includegraphics[scale=0.22]{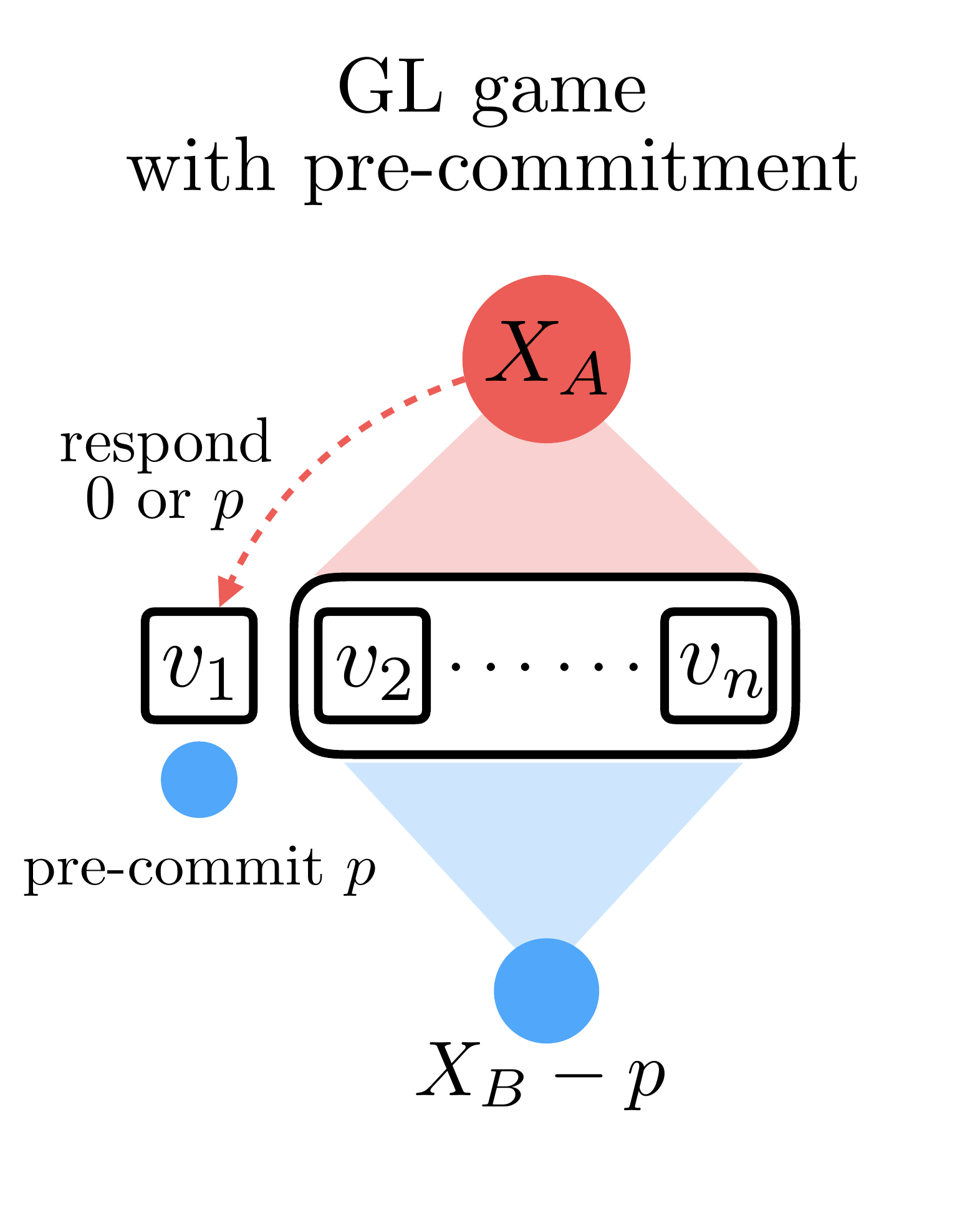}
		\caption{}\label{fig:GL_scenario}
	\end{subfigure} 
	\begin{subfigure}{.32\textwidth}
		\centering
		\includegraphics[scale=0.38]{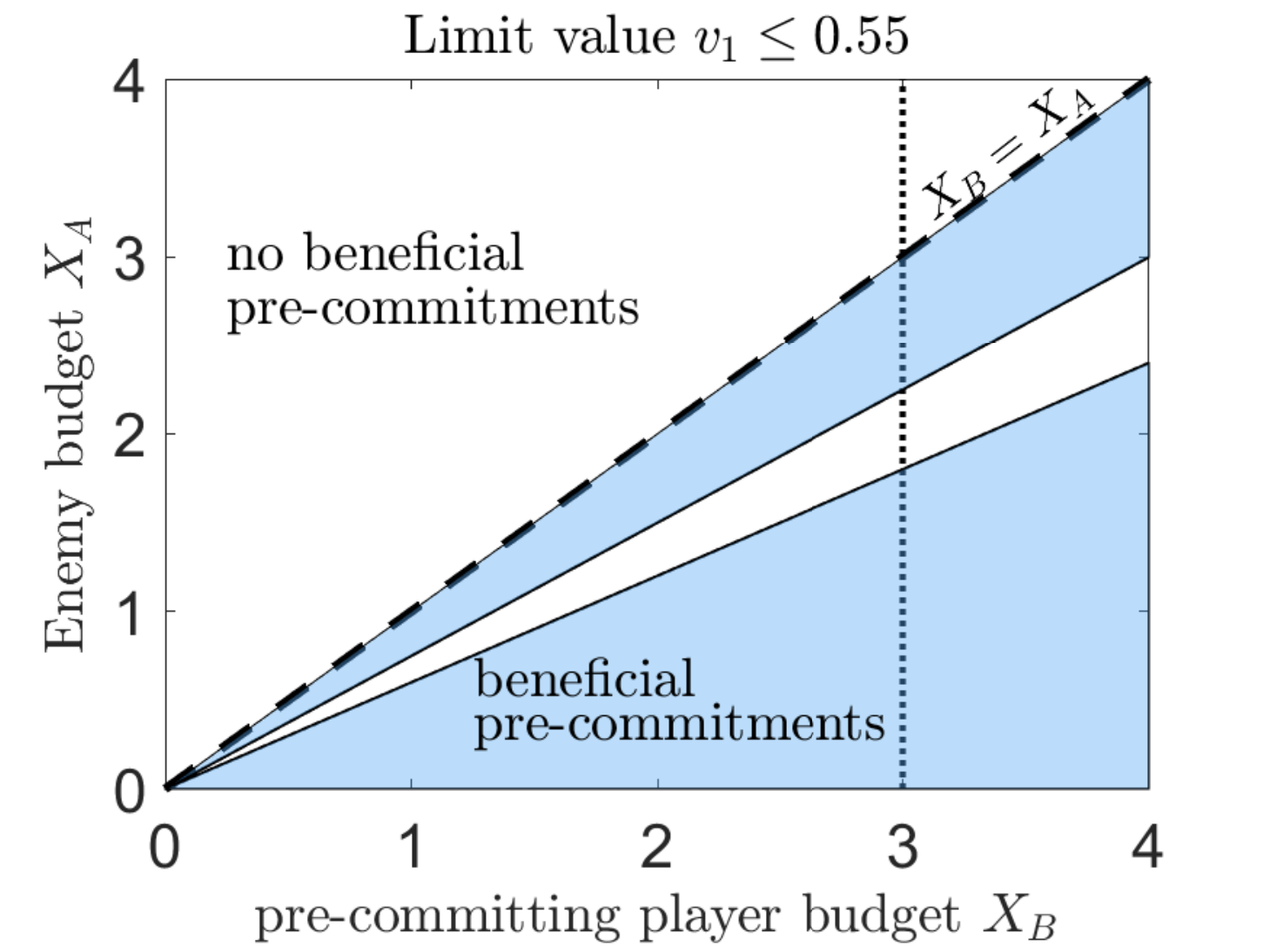}
		\caption{}\label{fig:GL_precommit_regions}
	\end{subfigure}  
	\begin{subfigure}{.32\textwidth}
		\centering
		\includegraphics[scale=0.38]{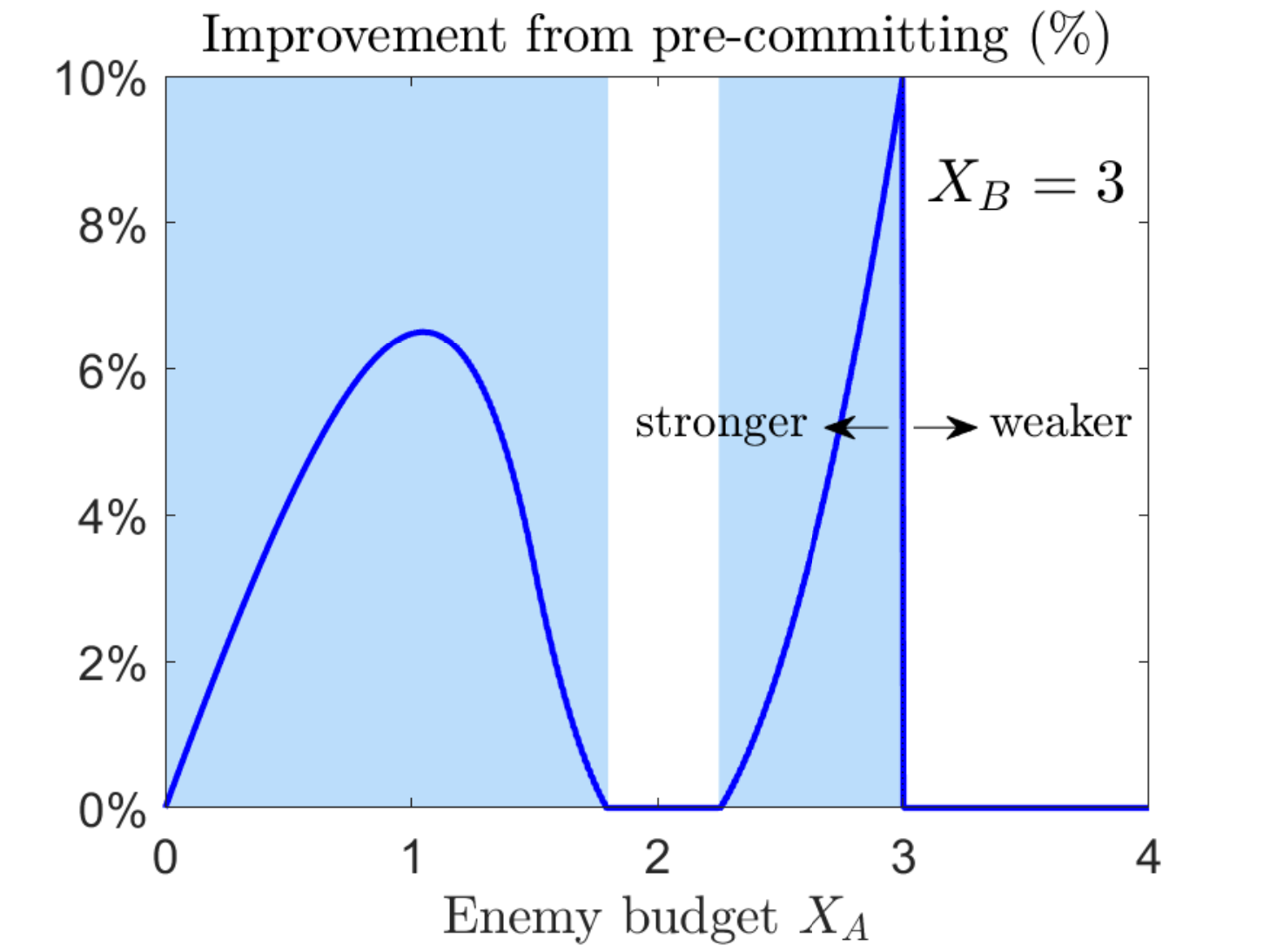}
		\caption{}\label{fig:GL_slice}
	\end{subfigure}
	%  Moreover, the value $v$ of the battlefield it pre-commits to must be sufficiently high.
	\caption{Scenario 1: Symmetric valuations. (a) Player $B$ pre-commits to a single battlefield of its choice. Our results establish that a pre-commitment to a single battlefield is preferable over pre-commitments to multiple battlefields of the same total value  (Lemma \ref{lem:single_t_optimal}). We thus seek necessary and sufficient conditions for which pre-committing to a single battlefield is beneficial for some set of valuations $\bs{v}$.
		(b) Parameter regimes where a player has an incentive to pre-commit (blue regions), i.e. to outperform its equilibrium payoff in the nominal GL game (Result 1). In this example, the value of the battlefield that $B$ can pre-commit to is restricted to not exceed 0.55. Here, we consider any set of valuations $\bs{v}$ whose total value is fixed to one. A full characterization of regions is given in Theorem \ref{thm:GL_result}. Incentives only exist when the pre-committing player is stronger than its opponent.  (c) The percent improvement in payoff that pre-commitments offer to player $B$ over playing the nominal GL game, as a function of the enemy's budget $X_A$ (tracing out parameters of the dotted vertical line in (b)).  Here, the pre-committing player's budget is fixed to $X_B=3$. Pre-committing can never outperform the nominal payoff when the player is weaker ($X_B < X_A$).}
	\label{fig:informal_thm1}
\end{figure*}

\subsection{A model of public pre-commitments}\label{sec:pc_model}

In this work, we evaluate whether there are advantages for a player to publicly pre-commit resources in competition. To do so, we consider an extension of General Lotto games, where one of the two players has the option to publicly pre-commit any amount of its resources to multiple battlefields. This interaction proceeds according to the following three-stage sequence of events.

\vspace{1mm}\noindent 1) One of the players, say $B$, selects a subset of battlefields $\mcal{P} \subseteq \mcal{B}$, and a pre-commitment, which is a tuple $\bs{p} = \{p_b\}_{b\in\mcal{P}}$ satisfying $\sum_{b\in\mcal{P}} p_b \leq X_B$. Here, $p_{b} \geq 0$ is the amount of resources placed on battlefield $b \in \mcal{P}$. Once pre-committed, player $B$ cannot add or take away the $p_b$ resources from each of the battlefields $b \in \mcal{P}$.

\vspace{1mm}\noindent 2) In response, the opponent ($A$) engages in a bidding contest on each of the pre-committed battlefields, where player $B$'s bid on battlefield $b\in\mcal{P}$ is deterministic and fixed at $p_b$. Player $A$ then has the opportunity to decide which of these battlefields to secure, and which ones to leave behind. Player $A$ secures battlefields by matching the pre-commitments with its own resources\footnote{In the bidding contests, we assume ties will be awarded to player $A$. This approximates (to arbitrary precision) the requirement that player $A$ sends $p_b + \epsilon$ in order to secure $b \in \mcal{P}$, for any $\epsilon > 0$. }. It withdraws entirely from the battlefields it decides not to match. We assume $A$'s response also becomes public knowledge as soon as it is taken.

\vspace{1mm}\noindent 3) The battlefields $A$ secures in stage 2 are awarded to player $A$, and the battlefields $A$ withdrew from are awarded to player $B$. No additional resources can be devoted to these battlefields. Both players subsequently engage in the game $\text{GL}(X_A-p_A,X_B-p,\{v_{A,b}\}_{b\notin\mcal{P}},\{v_{B,b}\}_{b\notin\mcal{P}})$ on the remaining set of battlefields $\mcal{B}\backslash\mcal{P}$ with their remaining resources. Here we denote $p = \sum_{b\in\mcal{P}} p_b$ as the total amount of resources $B$ used for pre-commitments, and $p_A \leq p$ as the amount of resources $A$ used to match pre-commitments in its response.

An illustration of this procedure\footnote{In our model, we are assuming in stage 2 that player $A$'s response to the pre-commitment becomes common knowledge before the subsequent Lotto game is played in stage 3. One can consider alternative formulations where such a response is not made public to the pre-committing player. Such analyses are outside the scope of the present paper, and we leave them as future research directions.
} is given in Figure \ref{fig:precommit_diagram}. A criteria for whether pre-commitments offer any advantages for player $B$ is to compare the  payoff it obtains, i.e. the value of battlefields $A$ withdraws from plus the payoff from engaging in the subsequent GL game in stage 3, to the payoff it would obtain if it did not have the option to pre-commit at all, i.e. its equilibrium payoff from the \emph{nominal} game $\text{GL}(X_A,X_B,\bs{v}_A,\bs{v}_B)$. Here, we are assuming player $A$ makes the optimal response to the pre-commitment, such that its subsequent payoff in stage 3 is maximized.

\subsection{Overview of contributions}

This paper analyzes scenarios where there is an advantage for a competitor to pre-commit resources. We study two distinct scenarios -- in the first scenario, the players are engaged in a General Lotto game where valuations are symmetric across players. In the second scenario, players are engaged in a General Lotto game where the battlefield valuations are asymmetric and relatively different across players.

% In both scenarios, we find there are instances where pre-committing offers benefits, though the conditions and types of benefits differ. Most notably, the weaker-resource player never has an incentive to pre-commit in the first scenario, whereas it can have incentives in the second scenario.

\vspace{1mm}\noindent{\bf Scenario 1 -- Symmetric battlefield valuations}: Here, the battlefield valuations are symmetric across players, i.e. $v_{A,b} = v_{B,b}$ for all $b\in\mcal{B}$. Player $B$ has the option to pre-commit (see Figure \ref{fig:precommit_diagram}) in the fashion described in Section \ref{sec:pc_model}. We are concerned with identifying conditions for which player $B$ has an incentive to pre-commit. Specifically, we seek necessary and sufficient conditions on the players' budgets $X_A$, $X_B$ for which there exists a beneficial pre-commitment for player $B$ on some vector $\bs{v}$ of battlefield valuations. 

Our analysis first establishes that these conditions can be narrowed to identifying whether a pre-commitment to a single battlefield can be beneficial. Indeed, we show a pre-commitment to a single battlefield of value $v$ is preferable over a pre-commitment to multiple battlefields whose total value is $v$ (Lemma \ref{lem:single_t_optimal}), given the cumulative value of the remaining battlefields in both scenarios are equivalent. The main highlight of our first set of results is given below, where we refer to the weaker (stronger) player as the one with the smaller (larger) resource endowment.
\begin{result}
	In the General Lotto game with symmetric valuations, player $B$ never has an incentive to pre-commit resources if it is the weaker player. However, player $B$ can have incentives to pre-commit if it is the stronger player.
\end{result}
Theorem \ref{thm:GL_result} identifies necessary and sufficient conditions on parameters for the existence of beneficial pre-commitments. These conditions also assert that the value of the battlefield(s) that player $B$ pre-commits to must exceed a certain threshold that depends on $X_A$ and $X_B$. Thus, if restrictions are placed on the total value that $B$ can pre-commit to, there may not exist incentives to pre-commit for any valuation vector $\bs{v}$. Note that when no such restrictions are placed, player $B$ can secure the entire set of battlefields through pre-commitments if $X_B > nX_A$, or if $X_B > X_A$ and there is only a single battlefield that is contested. 

\begin{figure*}[t]
	\centering
	\begin{subfigure}{.32\textwidth}
		\centering
		\includegraphics[scale=0.18]{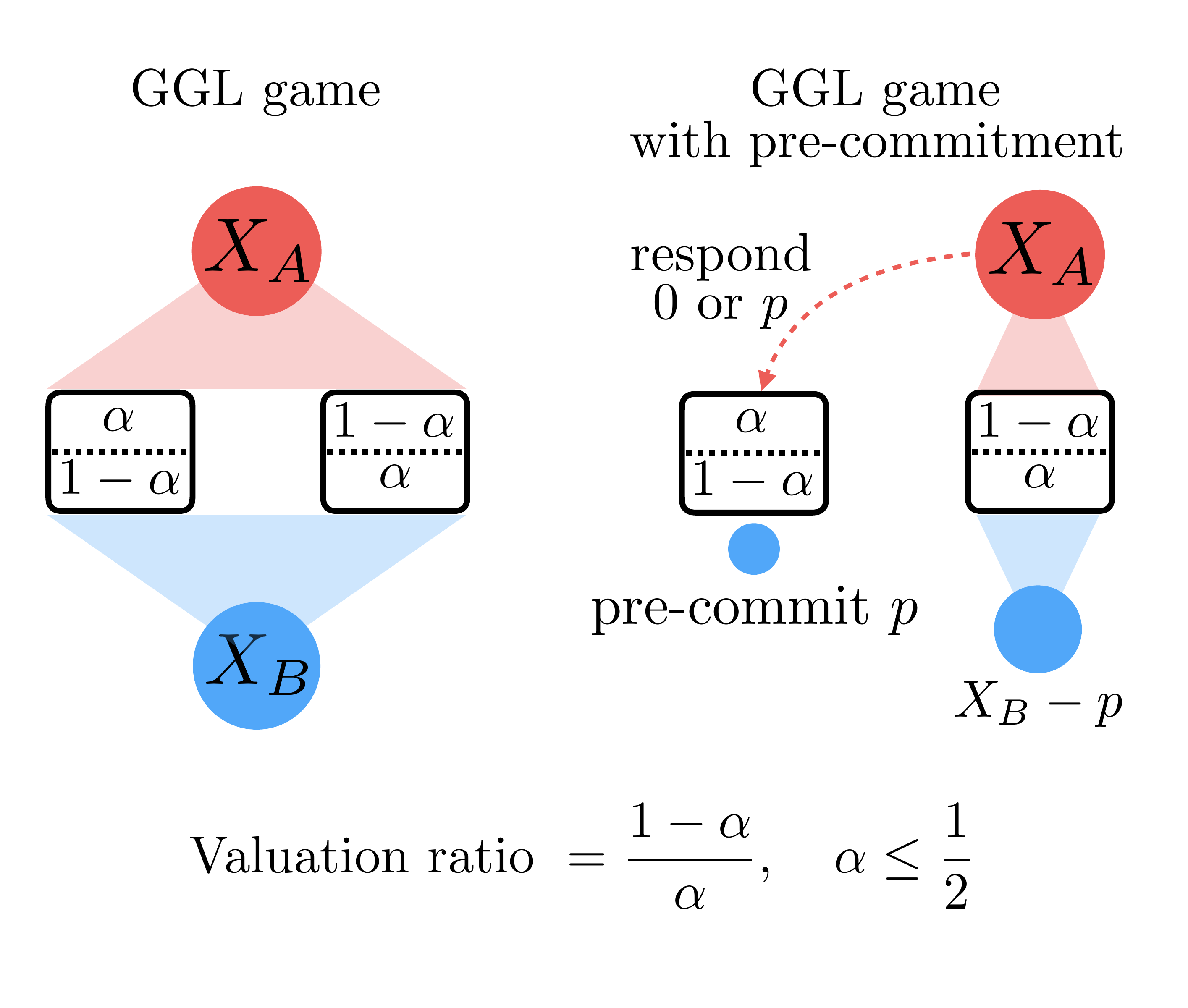}
		\caption{}\label{fig:GGL_scenario}
	\end{subfigure}
	\begin{subfigure}{.32\textwidth}
		\centering
		\includegraphics[scale=0.4]{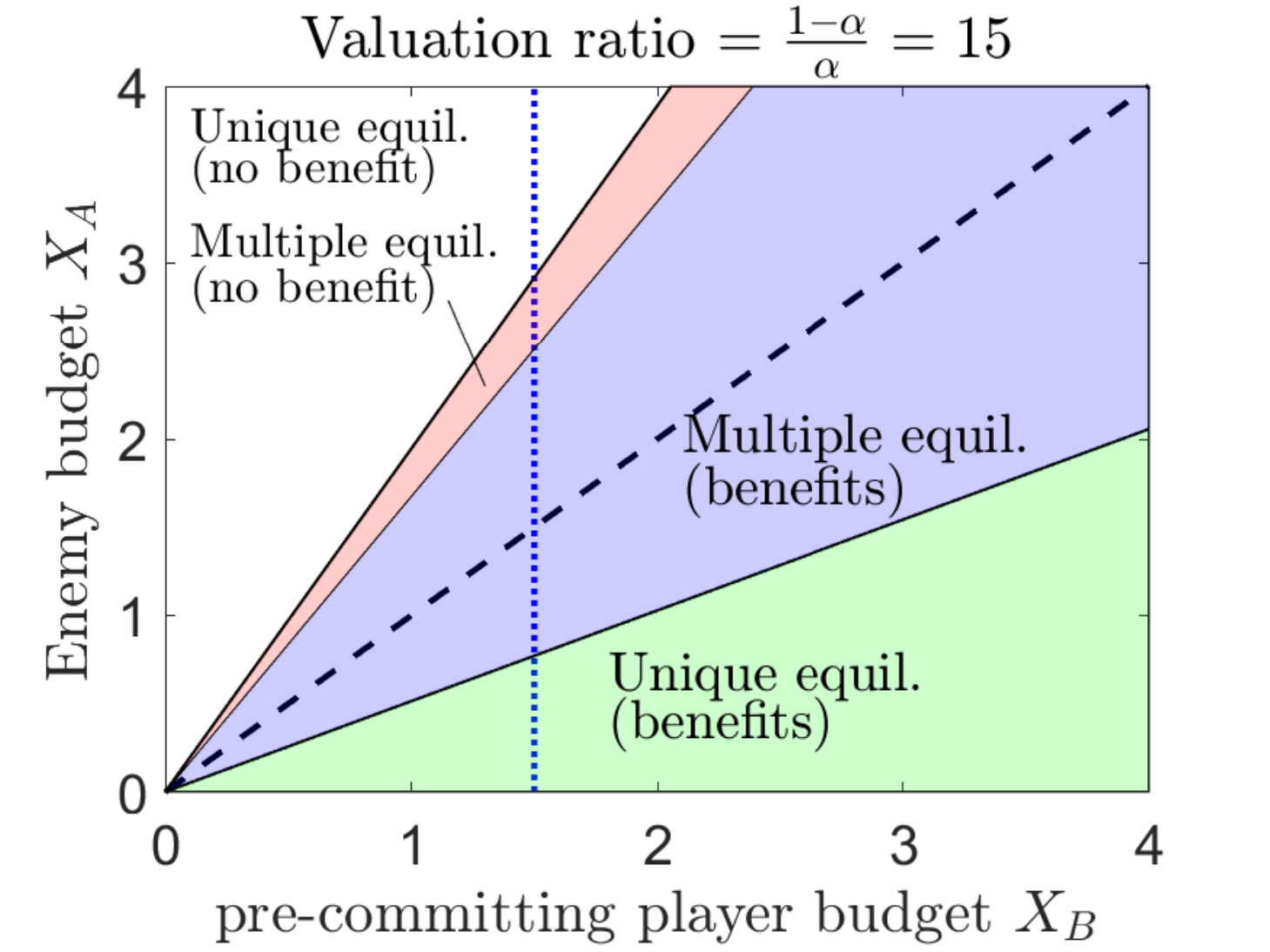}
		\caption{}\label{fig:GGL_precommit_regions}
	\end{subfigure}  
	\begin{subfigure}{.32\textwidth}
		\centering
		\includegraphics[scale=0.4]{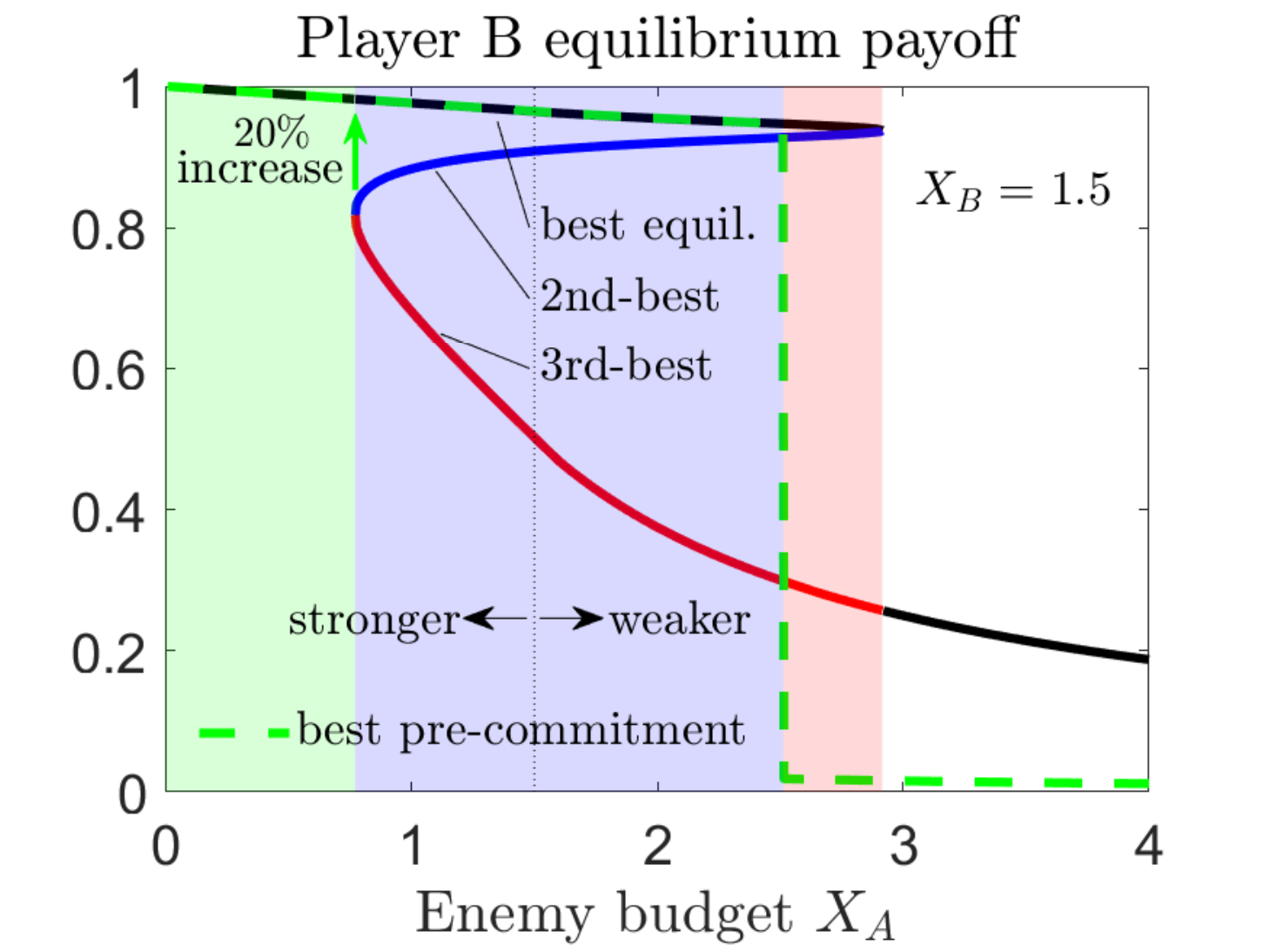}
		\caption{}\label{fig:GGL_equil_payoffs}
	\end{subfigure}
	\caption{Scenario 2: Asymmetric valuations. (a) The left diagram is the game $\text{GGL}(X_A,X_B,\alpha)$ under consideration. The right diagram shows the scenario where $B$ has the option to pre-commit resources to a battlefield. After the opponent's response, they play a GL game on the remaining battlefield.  (b) Parameter regions where a player in the GGL game has an incentive to pre-commit, based on two different criteria. The green region indicates parameters where there is a unique equilibrium payoff of the underlying GGL game, and there are beneficial pre-commitments that exceed this payoff. In the blue region, the GGL game admits multiple equilibria, and pre-commitments can exceed the second-highest equilibrium payoff (Theorem \ref{thm:second_best}). Multiple equilibria arise in the red region, but no such pre-commitments exist. 
		Observe that in this setting, incentives exist for a weaker player (blue region), whereas this was never the case for symmetric valuations.  (c) The equilibrium payoffs to player $B$ in the nominal (no pre-commitments) GGL game (traces out parameters of the dotted vertical line in (b)). The GGL game admits three equilibria in the blue and red regions, and the payoffs are ranked from best to worst for $B$. The dashed green line indicates the best payoff attainable from a pre-commitment. In this example, we note that it exceeds the best equilibrium payoff by an insignificant margin for $X_A \in [0,2.5]$ (approximately). However, it can improve upon the second-highest equilibrium payoff by up to 20$\%$.    }
	\label{fig:informal_thm2}
\end{figure*}

We also find that any beneficial pre-commitment forces player $A$ to withdraw, i.e. by pre-committing more than $X_A$ resources to a battlefield. Figure \ref{fig:GL_precommit_regions} depicts an example showing regions of the parameters $X_A,X_B$ where there exist beneficial pre-commitments for some valuation vector $\bs{v}$, under the restriction that the total value that $B$ pre-commits to does not exceed $0.55 \cdot \|\bs{v}\|_1$. Figure \ref{fig:GL_slice} shows the  improvement factor that an optimal pre-commitment offers, i.e. the best percent payoff gain over not pre-committing at all, over all possible valuations $\bs{v}$.  The main conclusion from the analysis of Scenario 1 is that a weaker player would never pre-commit resources.

% We find that it must pre-commit an amount of resources exceeding the weaker player's resource endowment. Moreover, the optimal pre-commitment is made on a single battlefield with sufficiently high value.

% pre-commitments are beneficial under the constraint that $B$ can only pre-commit to a battlefield of value up to $0.55$, where the total value of all battlefields is one. 

% Additional numerical studies are also provided in Section \ref{sec:GL_sims}. Here, we find that single battlefield pre-commitments outperform multiple battlefield pre-commitments even when the valuations $\bs{v}$ are drawn randomly. 

\vspace{1mm}\noindent{\bf Scenario 2 -- Asymmetric battlefield valuations}: In this scenario, the battlefield valuations are asymmetric and relatively different across players. We will focus on a particular two-battlefield setup where player $A$'s valuation of the first battlefield is $v_{A,1} = \alpha \in (0,1/2]$, and player $B$'s valuation is $v_{B,1} = 1-\alpha$. The valuations are reversed on the second battlefield, i.e. $v_{A,2} = 1-\alpha$ and $v_{B,2} = \alpha$ (Figure \ref{fig:GGL_scenario}). Valuations asymmetries in Lotto games were studied in \cite{Kovenock_2020}, where the authors termed this broader class of games as ``Generalized" GL games. We thus denote our particular settings concisely with $\text{GGL}(X_A,X_B,\alpha)$.

The equilibria of Generalized GL games have been characterized by Kovenock and Roberson \cite{Kovenock_2020}. In particular, multiple payoff-distinct equilibria can arise, as players' relative valuations of battlefields can differ. Such equilibria in general are difficult to express, however, because they correspond to the zeros of a piece-wise continuous cubic polynomial.

% We refer to the quantity $\frac{1-\alpha}{\alpha} \geq 1$ as the ``valuation ratio", or how much more a player values one battlefield over the opponent.

Because multiple payoff-distinct equilibria can arise in the underlying GGL game, the criteria to evaluate whether pre-commitments offer any benefit differs from the criteria in Scenario 1. Our results (details in Theorem \ref{thm:second_best}) identify necessary and sufficient conditions on parameters for which pre-committing can outperform the second-highest equilibrium payoff of $\text{GGL}(X_A,X_B,\alpha)$.
\begin{result}
	In the General Lotto game with asymmetric valuations, pre-commitments can offer improvement over the second-highest equilibrium payoff in instances that admit multiple payoff-distinct equilibria. These incentives exist when the pre-committing player is either weaker or stronger. Furthermore, there are instances admitting a unique equilibrium where pre-committing offers improvement for a weaker player.
\end{result}
A weaker player can have incentives to pre-commit when valuations are asymmetric, but it never has incentives when valuations are symmetric (Result 1). Because valuations are asymmetric across players, a pre-commitment to the $1-\alpha$ battlefield can ensure the opponent does not devote any resources to the more valuable battlefield. The opponent is more likely to withdraw from such a pre-commitment because it is on its non-priority battlefield worth only  $\alpha$. Though our results focus on a two battlefield setting, the intuition suggests that beneficial pre-commitments can exist for the same reasons in more general settings.

Figure \ref{fig:GGL_precommit_regions} illustrates parameter regions where pre-commitments offer benefits over the second-highest equilibrium payoff when multiple equilibria exist, and over the unique equilibrium when only one exists. Figure \ref{fig:GGL_equil_payoffs} shows equilibrium payoffs in the GGL game, and the relative benefits that pre-committing offers.  The improvement factor over the best equilibrium, however, is insignificant (usually up to a $1\%$ improvement). Nonetheless, this suggests that pre-committing can serve as a mechanism to avoid low equilibrium payoffs in the nominal GGL game. That is, one can guarantee a payoff close to the best equilibrium payoff by pre-committing, whereas there is no such guarantee if players engage in the nominal GGL game.

\subsection{Related literature}

A primary line of research in Colonel Blotto games focuses on characterizing its equilibria. Since Borel's initial study \cite{Borel}, many works have advanced this thread over the last one hundred years \cite{Gross_1950,Roberson_2006,Schwartz_2014,Macdonell_2015,Thomas_2018,Kovenock_2020,Boix_2020}. As such, there are several variants of the Colonel Blotto game. The General Lotto game has been studied extensively \cite{Bell_1980,Myerson_1993,Hart_2008,Kovenock_2020}, and equilibria can be characterized for all instances. Due to its tractability, the General Lotto game, as well as other variants, are often adopted to study more complex adversarial environments. For instance, they have been used in engineering domains such as network security \cite{Fuchs_2012,Shahrivar_2014,Guan_2019} and the security of cyber-physical systems \cite{Gupta_2014a,Ferdowsi_2020}.

%The difficulty is that mixed strategies must be valid $n$-variate joint distributions  on each player's feasible allocations, i.e. the support is confined to a simplex.

Informational elements in Colonel Blotto and General Lotto games constitute a significant theme in the literature. Our work is closest to an area that concerns similarly sequenced Blotto and Lotto games. In particular, \cite{Kovenock_2012} introduced a two-vs-one model, and identified when a (public) unilateral transfer of resources between coalitional players is beneficial. Subsequent work in \cite{Gupta_2014a,Gupta_2014b} considers similar settings where the two players can decide to add battlefields in addition to transferring resources amongst each other. Counter-intuitively, they showed that the players in the coalition achieve better performance if the transfers are made public to their adversary. Previous studies focused on settings where endogenously adding battlefields is costly \cite{Kovenock_2010_KER}. Pre-committing resources was considered in \cite{Vu_2021_favoritism}, but in a different context that involves favoritism. There, they studied a one-shot Blotto game with resources that are pre-allocated non-strategically to the battlefields. 

Incomplete information settings pertaining to payoff-relevant parameters in zero-sum games has received attention in recent years. The computation of security strategies in repeated, asymmetric information zero-sum games characterizes the balance of revealing private information through one's actions over time \cite{Li_2019,Li_2020,Kartik_2021asymmetric}. In Blotto and Lotto games, several papers study one-shot settings where players have uncertainty about each other's resource budgets. In \cite{Adamo_2009}, the players'  budgets are random variables drawn from a common distribution, and each player holds private information only about their own budget. Another recent formulation characterizes asymmetric Bayes-Nash equilibria when one player does not have knowledge about the opponent's budget, but its own budget is common knowledge \cite{Paarporn_2021_budget}.  Uncertainty about battlefield valuations has also been featured in the literature \cite{Kovenock_2011,Paarporn_2019,Ewerhart_2021}.

% As such, symmetric Bayes-Nash equilibria are identified. The model of \cite{Kim_2017} provides Bayes-Nash equilibria to a Lottery Blotto game, a variant that generally admits pure strategy Nash equilibria \cite{Robson_2005,Vu_2020thesis}.

%%%%%%%%%%%%%%%%%%%%%%%%%%%%%%%%%%%%%%%%%%%%%%%%%%%%%%%
% \section{Preliminaries}
% \input{preliminaries}

%%%%%%%%%%%%%%%%%%%%%%%%%%%%%%%%%%%%%%%%%%%%%%%%%%%%%%%
\section{Pre-commitments with symmetric valuations}
% !TEX root = main.tex

In this section, we seek conditions on the players' resource budgets $X_A,X_B$ for which one of the players has incentives to pre-commit resources in GL games with symmetric valuations. We begin with a basic primer on their equilibria.

\subsection{Primer on equilibria in GL games}

The two-player General Lotto game $\text{GL}(X_A,X_B,\bs{v})$ always admits a unique equilibrium payoff, which is given as follows. 
\begin{fact}
	The equilibrium payoff to player $i \in \{A,B\}$ in the game ${\rm GL}(X_A,X_B,\bs{v})$ is $\pi_i^\texttt{nom} := \phi \cdot L(X_i, X_{-i})$, where 
	\begin{equation} \label{eq:twoplayer_payoff}
	L(X_i, X_{-i}) := 
	\begin{cases} 
	\frac{X_i}{2X_{-i}}& \quad \text{if } 0 < X_i \leq X_{-i} \\
	1 - \frac{X_{-i}}{2X_i}& \quad \text{if } X_i > X_{-i}.
	\end{cases} 
	\end{equation}
	and $\phi := \sum_{b\in\mcal{B}} v_b$. The equilibrium payoff to player $-i$ is $\pi_{-i}^\texttt{nom} = \phi(1 - L(X_i,X_{-i}) )$.
\end{fact}
An admissible strategy for $i \in \{A,B\}$ is any $n$-variate distribution $F_i$ over $\mathbb{R}_+^n$ that satisfies the budget constraint in expectation \eqref{eq:lotto_constraint}. The equilibrium strategies for the GL game are detailed below.
\begin{fact}[GL equilibrium strategies \cite{Hart_2008,Kovenock_2020}]\label{fact:GL_strats}
	Suppose player $i \in \{A,B\}$ is the stronger player, i.e. $X_i \geq X_{-i}$. Then the players' equilibrium marginal (cumulative) distribution on resource allocation for battlefield $b \in \mcal{B}$ is given by
	\begin{equation}
	\begin{aligned}
	F_{i,b}(x) &= \frac{x}{2X_i v_b}, \quad x \in [0,2X_i v_b) \\
	F_{-i,b}(x) &= 1-\frac{X_{-i}}{X_i} + \frac{X_{-i}}{2X_i^2 v_b}x, \quad x \in [0,2X_i v_b) \\
	\end{aligned}
	\end{equation}
\end{fact}

% Because of the expectation constraint, one only needs to specify the marginal distributions of any strategy, i.e. the $n$ univariate distributions $\{F_{i,b}\}_{b\in\mcal{B}}$ that determine how many resources are allocated to each battlefield. This is unlike Colonel Blotto games, however, which requires the $n$-variate joint distribution to have its support confined to the set $\Delta_n(X_i)$ for player $i$.

The stronger player uses uniform distributions on allocations for each battlefield, and the weaker player uses uniform distributions with the same support, but combined with a point mass at zero. In essence, the stronger player competes more aggressively on each battlefield because the weaker player ``gives up" on each battlefield independently with probability $1-\frac{X_{-i}}{X_i}$.

In the forthcoming analysis on pre-commitments in GL games, we will not need to make explicit use of the details of Fact \ref{fact:GL_strats} since we are concerned only with the associated payoffs in the sequence of events outlined in Section \ref{sec:pc_model}. Our results thus rely heavily on the characterization of equilibrium payoffs in  \eqref{eq:twoplayer_payoff}.

\subsection{Beneficial pre-commitments in GL games}

% Let us consider the scenario where player $B$ selects a subset of battlefields $\mcal{P}\subseteq \mcal{B}$ and a pre-commitment $\bs{p} = \{p_b\}_{b\in\mcal{P}} \in \mbb{R}_+^{|\mcal{P}|}$, where $0 \leq \sum_{b \in \mcal{P}} p_b \leq X_B$

Recall the pre-commitment model of Section \ref{sec:pc_model} (depicted in Figure \ref{fig:precommit_diagram}). Here, we assume that any valuation vector $\bs{v}$ has a fixed total value $\phi=\sum_{b\in\mcal{B}} v_b$, and we denote $\mcal{V}$ as the set of all such valuation vectors. Suppose there is also a restriction $\sum_{b\in\mcal{P}} v_b \leq \bar{v}$ imposed on the total value of the battlefields $\mcal{P}$ that $B$ can pre-commit to, where $\bar{v}\in [0,\phi]$ is a fixed and exogenous parameter. We will refer to $\bar{v}$ as the \emph{limit value}.

\begin{figure}[t]
	\centering
	\begin{subfigure}{.32\textwidth}
		\centering
		\includegraphics[scale=.25]{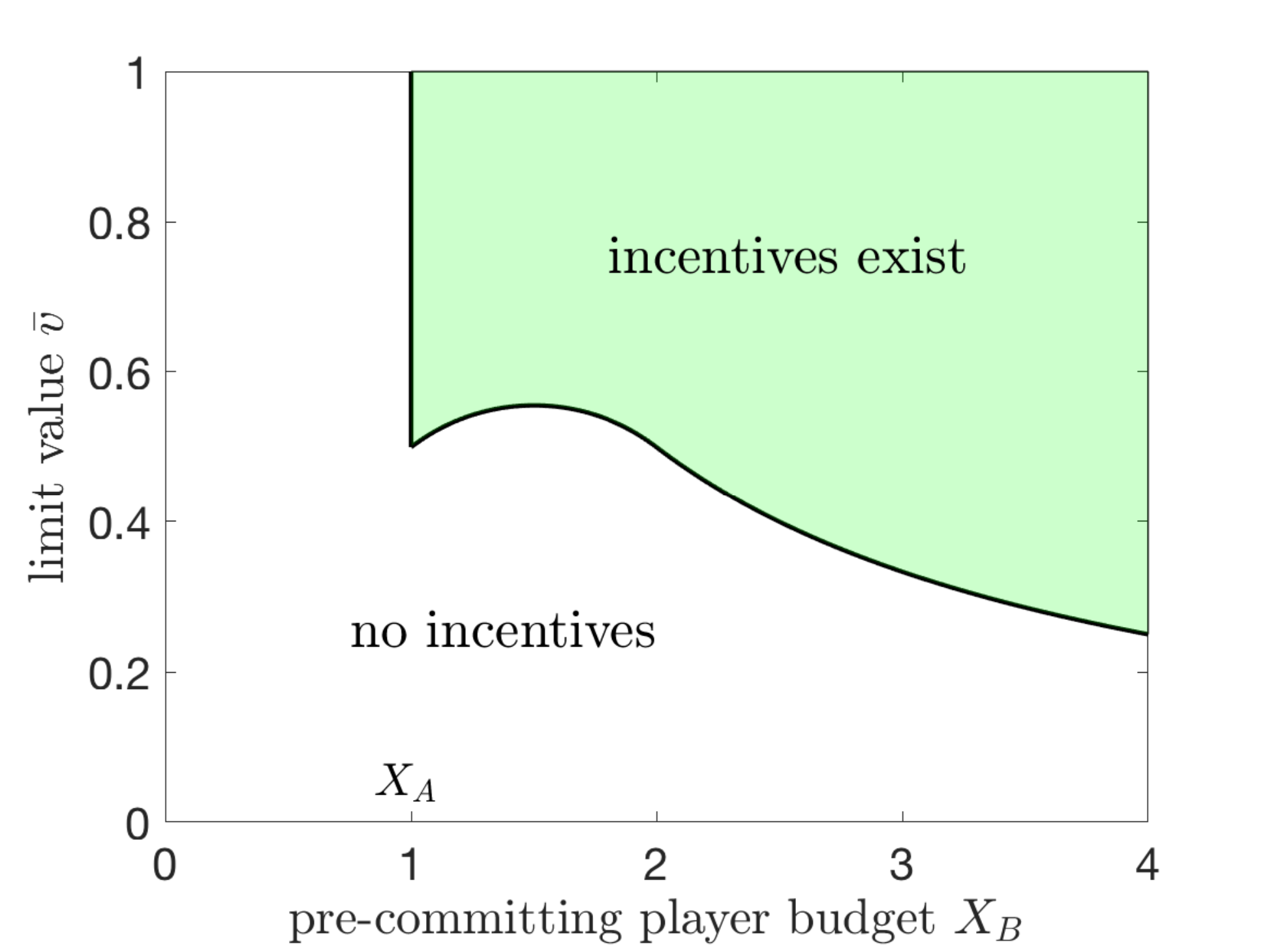}
		\caption{ }\label{fig:GL_full_regions}
	\end{subfigure}
	\begin{subfigure}{.32\textwidth}
		\centering
		\includegraphics[scale=.25]{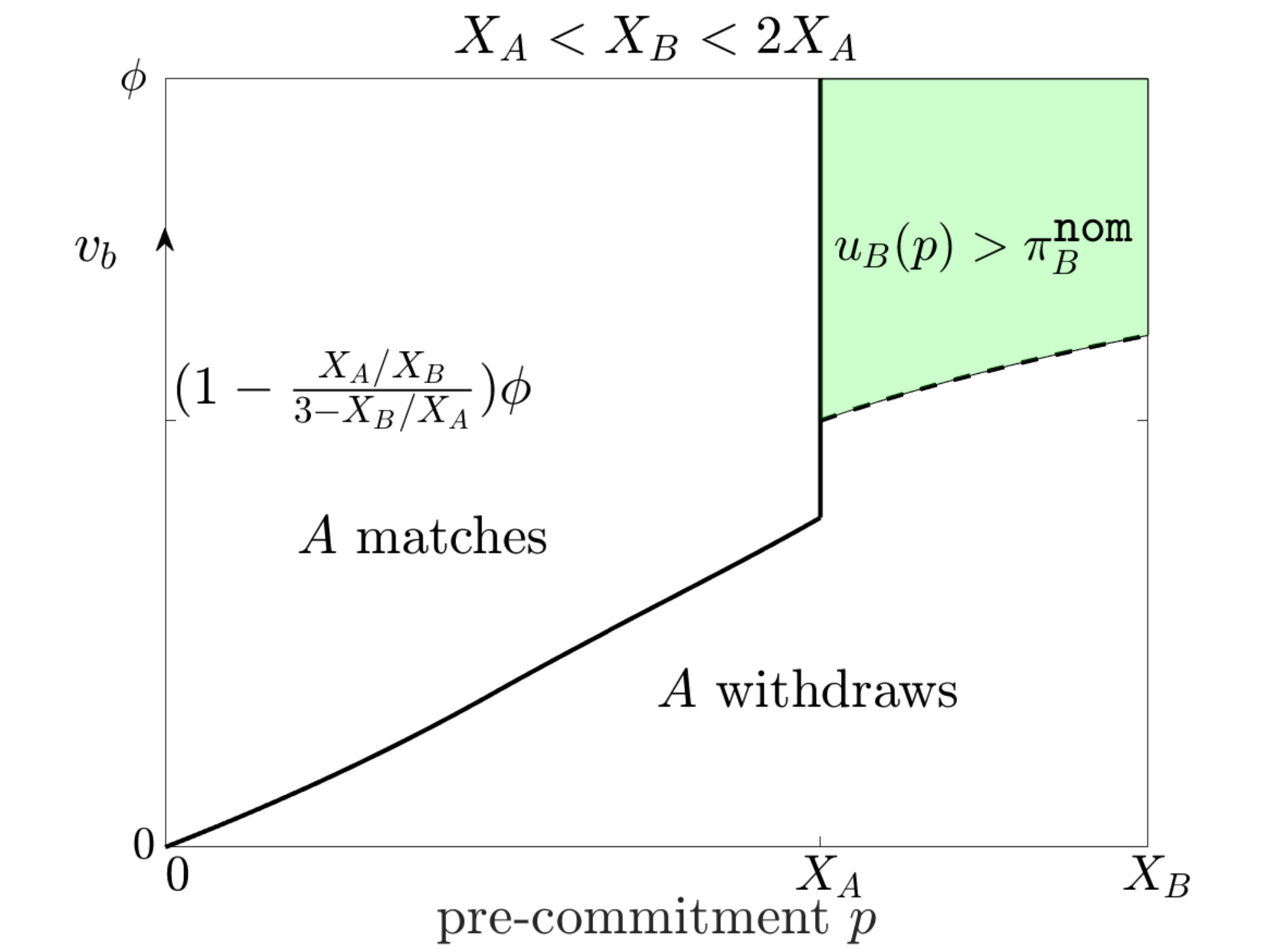}
		\caption{}\label{fig:GL_t_regions1}
	\end{subfigure}  
	\caption{(a) Full parameter region (in green) where there exists beneficial pre-commitments in $\text{GL}(X_A,X_B,\bs{v})$ (from Theorem \ref{thm:GL_result}). Here, we set $\phi=1$. (b,c) These diagrams depict regions of player A's optimal response against any pre-commitment $p \in [0,X_B]$ on a single battlefield of value $v_b \in [0,\phi]$ (suppose $\bar{v} = \phi$ here), and the total value of battlefields is $\phi$. (b) When $X_A<X_B<2X_A$, there is an incentive to pre-commit to a single battlefield $b$ if and only if $v_b > (1 - \frac{X_A/X_B}{3-X_B/X_A})\phi$ (green region). The set of beneficial pre-commitments are shown as the shaded green region. The exact analytical characterizations of the border lines are given in the proof of Theorem \ref{thm:GL_result} (Appendix).}
	\label{fig:GL_stronger_regions}
\end{figure}

% subject to a value limit $\sum_{b\in\mcal{P}} v_b \leq \bar{v}$. Here, $\bar{v} \in [0,\phi]$ is an exogenous parameter that restricts the total value of the battlefields that $B$ is able to pre-commit to. We will thus consider the game instances where the battlefield valuations $\bs{v} \in \mbb{R}_+^n$ have fixed total value, $\sum_{b\in\mcal{B}} v_b = \phi$.

Player $A$ observes the pre-commitment strategy $\bs{p}$. For each battlefield $b\in\mcal{P}$, $A$ must decide either to \emph{match} the pre-commitment with $p_b$ of its own resources -- thus securing the value $v_b$ for itself -- or to \emph{withdraw} from battlefield $b$ entirely.  After $A$ makes a response for each battlefield $b\in\mcal{P}$, both players then engage in a GL game with their remaining resources over the remaining set of battlefields $\mcal{B}\backslash\mcal{P}$. Specifically, if player $A$ decides to match pre-commitments on battlefields $\mcal{M} \subseteq \mcal{P}$, it obtains the payoff
\begin{equation}\label{eq:A_match}
\begin{aligned}
u_{A,\mcal{M}}(\bs{p};\bs{v}) := v_\mcal{M} + (\phi-v_\mcal{M})\cdot L(X_A-p_\mcal{M},X_B-p_\mcal{P})
\end{aligned}
\end{equation}
where we henceforth denote $v_\mcal{S} = \sum_{b\in\mcal{S}} v_b$ and $p_\mcal{S} = \sum_{b\in\mcal{S}} p_b$ for any subset $\mcal{S} \subseteq \mcal{B}$. If $\mcal{M} = \varnothing$, then $v_\mcal{M} = 0$ and $p_\mcal{M} = 0$. If $p_b > X_A$, then $A$ cannot match the pre-commitment and is forced to withdraw from battlefield $b$. We assume player $A$ makes an optimal response to $\bs{p}$ and obtains the subsequent payoff
\begin{equation}\label{eq:A_twoPlayer_BR}
u_A(\bs{p};\bs{v}) := \max_{\mcal{M}\subseteq\mcal{P}}\left\{ u_{A,\mcal{M}}(\bs{p};\bs{v}) : p_\mcal{M} \leq X_A \right\} .
\end{equation}
Player $B$ then obtains the subsequent payoff $u_B(\bs{p};\bs{v}) := \phi - u_A(\bs{p};\bs{v})$. Alternatively, player $B$ has the option to not pre-commit at all, wherein both players simply engage in the \emph{nominal}  game $\text{GL}(X_A,X_B,\bs{v})$ in simultaneous play over the full set of battlefields $\mcal{B}$. Under this option, the players derive the equilibrium payoffs $\pi_i^{\texttt{nom}} = \phi\cdot L(X_i,X_{-i})$ from \eqref{eq:twoplayer_payoff}.

We are interested in identifying conditions on the budgets $X_A,X_B$ and on the limit value $\bar{v}$ for which player $B$ has an incentive to pre-commit. A formal definition of when an incentive exists is given below. 
\begin{definition}
	We say player $B$ has an \emph{incentive to pre-commit} if there exists a subset $\mcal{P}\subseteq \mcal{B}$ satisfying $\sum_{b\in\mcal{P}} v_b \leq \bar{v}$, and a pre-commitment $\bs{p}$  such that $u_B(\bs{p};\bs{v}) > \pi_B^{\texttt{nom}}$ for some set of battlefield valuations $\bs{v} \in \mcal{V}$.
\end{definition}
The following lemma establishes that player $B$ weakly prefers to pre-commit to a single high-value battlefield than to many battlefields whose total value is the same as the single battlefield.

\begin{figure*}
	\centering
	\begin{subfigure}{.32\textwidth}
		\centering
		\includegraphics[scale=0.18]{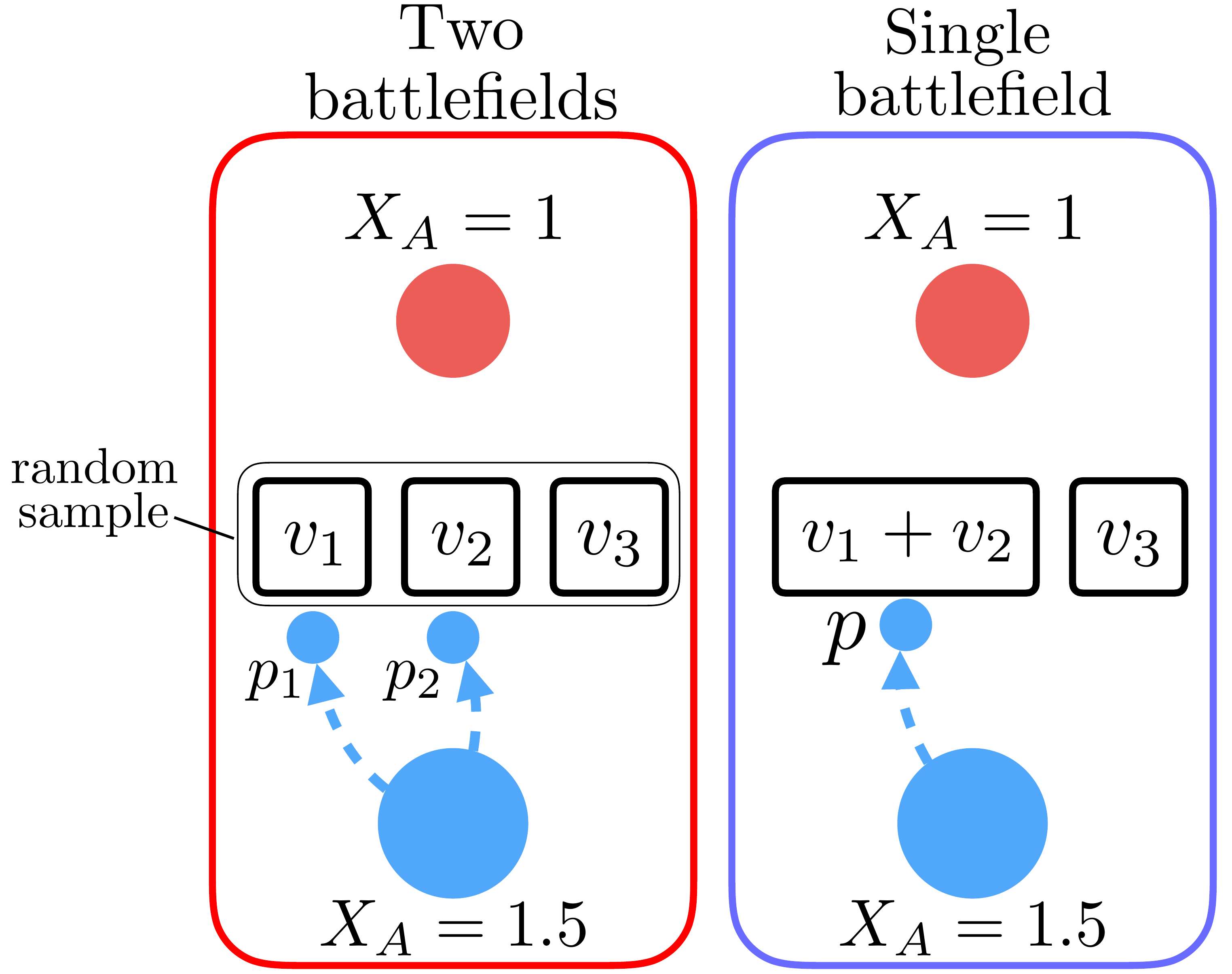}
		\caption{}\label{fig:sim_scenario}
	\end{subfigure}
	\begin{subfigure}{.32\textwidth}
		\centering
		\includegraphics[scale=0.3]{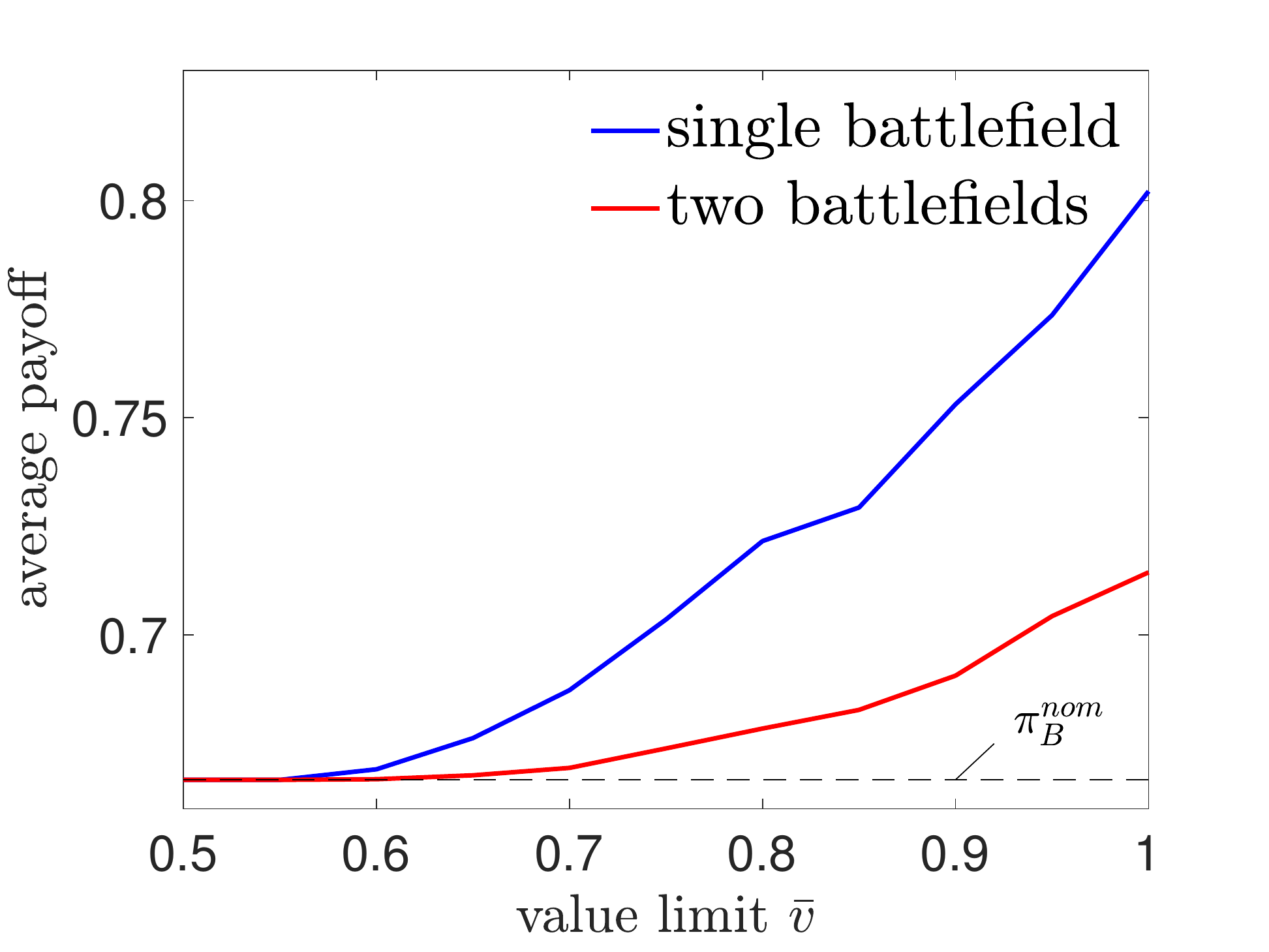}
		\caption{}\label{fig:GL_sim}
	\end{subfigure}
	\begin{subfigure}{.32\textwidth}
		\centering
		\includegraphics[scale=0.3]{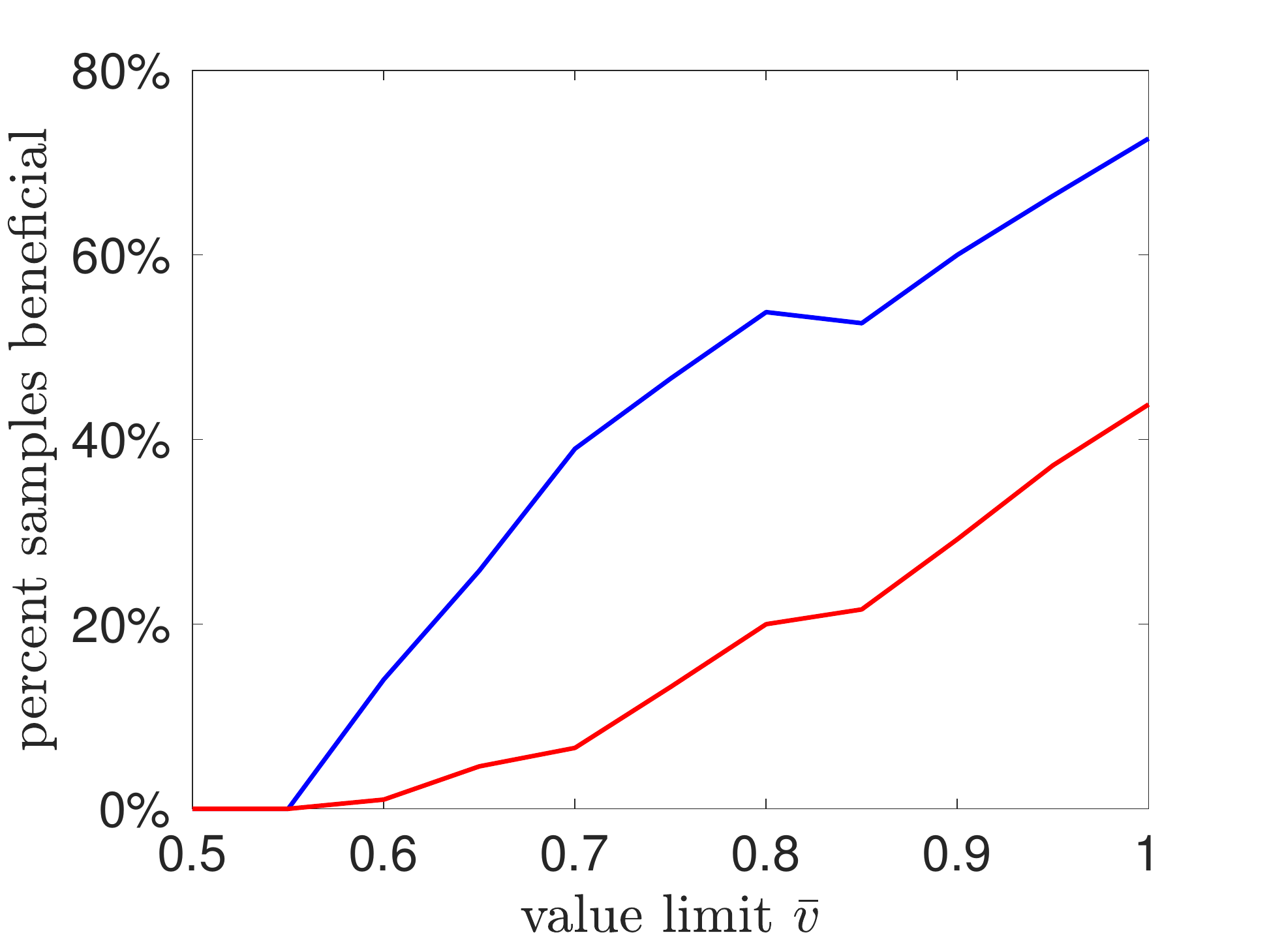}
		\caption{}\label{fig:GL_sim_percent_benefit}
	\end{subfigure}
	\caption{(a) Diagram of our simulation setup. First, a random set of valuations $\bs{v} \in \mcal{V}$ ($n=3$) is drawn from the uniform distribution on $\mcal{V}$. We calculate the optimal two-battlefield pre-commitment to the first two battlefields. We then calculated the optimal pre-commitment to a single battlefield of value $v_1+v_2$. Here, $X_B = 1.5$ and $X_A = 1$. (b) We performed calculations on 500 independent samples of $\bs{v}$ (for each $\bar{v}$ ranging from 0:0.05:1), with $\phi = 1$ and $n=3$. This plot shows the average payoff obtained from pre-committing to a single battlefield (blue) vs pre-committing to two battlefields (red). (c) This plot shows the percent of samples for which the optimal pre-commitment provided benefits over the nominal payoff. Note that no benefits are available for any $\bar{v} < 5/9$ (Theorem \ref{thm:GL_result}). }
	\label{fig:GL_sims}
\end{figure*}

\begin{lemma}\label{lem:single_t_optimal}
	Consider any instance of battlefield valuations $\bs{v}\in\mcal{V}$.
	For any pre-commitment $\bs{p}$ on a subset $\mcal{P}\subseteq \mcal{B}$ of battlefields (satisfying $v_\mcal{P} \leq \bar{v}$), there is a pre-commitment $p'\in [0,X_B]$ to a single battlefield of value $v_\mcal{P}$ that belongs to a corresponding set of valuations $\bs{v}' \in \mcal{V}$ such that
	\begin{equation}
	u_B(\bs{p};\bs{v}) \leq u_B(p';\bs{v}').
	\end{equation}
\end{lemma}
\begin{proof}
	If player $A$'s optimal response is to match or withdraw on all battlefields $b \in \mcal{P}$, then consider an instance $\bs{v}'\in\mcal{V}$ where one of its battlefields has value $v_\mcal{P}$. The  pre-commitment $p_\mcal{P}$ to this single battlefield elicits the same response and hence identical payoffs. If $A$'s optimal response is not to match or withdraw on all $b \in \mcal{P}$, then we can focus exclusively on the case $|\mcal{P}|=2$. To see this, suppose player $A$'s optimal response is to match on a subset $\mcal{M} \subset \mcal{P}$. Consider another instance of battlefields $\bs{v}'$ where $v_1' = v_\mcal{M}$ and $v_2' = v_{\mcal{P}\backslash \mcal{M}}$, and a pre-commitment $(p_1',p_2') = (p_\mcal{M},p_{\mcal{P}\backslash \mcal{M}})$. Player $A$'s optimal response is to match on $v_1'$ and withdraw on $v_2'$, giving the same subsequent payoff $u_B(\bs{p};\bs{v})$.
	
	So, let us consider a pre-commitment $\bs{p} = (p_1,p_2)$ to $\bs{v}$. Suppose $A$'s optimal response is to match $v_1$ and to withdraw from $v_2$ (w.o.l.o.g). First, let us assume that $p_1 + p_2 > X_A$, so that $A$ cannot match on both battlefields. Here, $B$ is necessarily the stronger player. From optimality of $A$'s  response, it holds that 
	\begin{equation}\label{eq:callfold_ineq}
	(\phi-v_1-v_2)\cdot L(X_A,X_B-p_1-p_2) \leq u_{A}(\bs{p};\bs{v})
	\end{equation}
	In words, matching only on $v_1$ outperforms withdrawing from both $v_1$ and $v_2$. Let us consider an instance with a set of battlefields $\bs{v}'$, where $v'_1 = v_1+v_2$. The pre-commitment $p' = p_1+p_2$ only on the first battlefield forces player $A$ to withdraw from $v'_1$, since $p_1+p_2 > X_A$. The subsequent payoff is $u_A(p';\bs{v}') = (\phi-v_1-v_2)\cdot L(X_A,X_B-p_1-p_2)$, and from \eqref{eq:callfold_ineq} we obtain
	\begin{equation}
	u_B(\bs{p};\bs{v}) \leq u_B(p';\bs{v}') .
	\end{equation}
	Now, let us assume that $p_1+p_2 \leq X_A$, and suppose $A$'s optimal response is to match $v_1$ and to withdraw from $v_2$. Let us define $f(s) = u_{A,\{1\}}(p_1,s)$ and $g(s) = u_{A,\{1,2\}}(p_1,s)$ for $s \geq 0$. The value $f(s)$ is the payoff $A$ obtains from matching $v_1$ and withdrawing from $v_2$ against the pre-commitment $(p_1,s)$. It is a strictly increasing function of $s$. The value $g(s)$ is the payoff $A$ obtains from matching both pre-commitments, and is a strictly decreasing function of $s$. Note that $f$ and $g$ are continuous in $s$, and that $g(0) > f(0)$. By the assumption of $A$'s optimal response to $\bs{p} = (p_1,p_2)$, we have $f(p_2) \geq g(p_2)$. Also, from the assumption that $p_2 < X_A$, the function $\max\{f(s),g(s)\}$ is continuous on $s \in [0,p_2]$. Hence, there exists a unique pre-commitment $p_2^* \in [0,p_2]$ for which $f(p_2^*) = g(p_2^*)$, making $A$ indifferent between matching and withdrawing on the second battlefield. Since $p_2^* \leq p_2$, we have
	\begin{equation}\label{eq:inequality_seq}
	u_A(p_1,p_2^*;\bs{v}) = f(p_2^*) \leq f(p_2) = u_A(p_1,p_2;\bs{v})
	\end{equation}
	and we obtain $u_B(p_1,p_2^*;\bs{v}) \geq u_B(p_1,p_2;\bs{v})$. In words, player $B$ can weakly improve its payoff by lowering its pre-commitment to the second battlefield down to $p_2^*$, whereupon an optimal response for $A$ is to match both pre-commitments. Let us consider an instance with a set of battlefields $\bs{v}'$ where $v_1' = v_1+v_2$, and a pre-commitment $p' = p_1+p_2^*$ only on $v_1'$. It follows that player $A$'s optimal response is to match, where it derives the payoff $u_A(p';\bs{v}') = u_A(p_1,p_2^*;\bs{v})$. From \eqref{eq:inequality_seq}, we obtain $u_B(p_1,p_2;\bs{v}) \leq u_B(p';\bs{v}')$.
\end{proof}
Lemma \ref{lem:single_t_optimal} implies that if there are no incentives to pre-commit to a single battlefield of any value up to the limit $\bar{v}$ on any set of valuations $\bs{v} \in \mcal{V}$, then there is no incentive to pre-commit to multiple battlefields with total value up to $\bar{v}$. In seeking conditions that give player $B$ incentives to pre-commit, we can thus simplify the analysis by narrowing the focus to single battlefield pre-commitments.  The following result thus provides necessary and sufficient conditions on $X_A,X_B$, and $\bar{v}$ for which incentives to pre-commit exist on some valuation vector $\bs{v} \in \mcal{V}$.
\begin{theorem}\label{thm:GL_result}
	If $X_B < X_A$, then there is never an incentive for player $B$ to pre-commit resources. If $X_A \leq X_B < 2X_A$, then  player $B$ has incentives to pre-commit resources  
	if and only if $\bar{v} > (1 - \frac{X_A/X_B}{3-X_B/X_A})\phi $. If $X_B \geq 2X_A$, then player $B$ has incentives to pre-commit resources if and only if $\bar{v} > \frac{X_A}{X_B}\phi $. Moreover, any beneficial pre-commitment $p$ to a single battlefield satisfies $p > X_A$.
\end{theorem}
The proof is provided in the Appendix. This result places necessary and sufficient conditions on the limit value $\bar{v}$ for which there exist instances that offer incentives to pre-commit. That is, a stronger player has incentives to pre-commit, provided that the value of the battlefield it pre-commits to is sufficiently high and is within the limit value $\bar{v}$. Note that this result does not immediately imply that there exist beneficial pre-commitments to an arbitrary subset of battlefields whose total value satisfies this condition. Figure \ref{fig:GL_stronger_regions} illustrates the parameter region where pre-commitments are beneficial. Interestingly, the minimum value of $v_b$ for which there exist beneficial pre-commitments is not monotonic in the budget ratio $\gamma = X_B/X_A$. It is increasing on $\gamma \in (1,\frac{3}{2})$, taking a maximum value of $\frac{5}{9}\phi$ at $\gamma = \frac{3}{2}$, and decreasing on $\gamma \in (\frac{3}{2},\infty)$ (Figure \ref{fig:GL_full_regions}). Moreover, the minimum $v_b$ is at least $\phi/2$ when $X_A < X_B < 2X_A$. Player $B$ must pre-commit an amount of resources that exceeds player A's budget in order to benefit.

\subsection{Discussion and numerical studies}\label{sec:GL_sims}

%The result of Theorem \ref{thm:GL_result} specifies necessary and sufficient conditions on budgets and the limit value $\bar{v}$ for which one can find an instance where player $B$ has an incentive to pre-commit to a single battlefield. If these conditions are not satisfied, then incentives to pre-commit to multiple battlefields also do not exist.

Even if the condition on the limit value $\bar{v}$ in Theorem \ref{thm:GL_result} is satisfied, there may not exist any beneficial pre-commitment on a particular valuation vector $\bs{v}$. Sometimes, pre-committing to multiple battlefields is preferable over pre-committing to a single battlefield. For example, consider three battlefields each with value $1/3$, a limit value of $\bar{v} = 1$, $X_B > 3$, and $X_A = 1$. A pre-commitment that places $p_b = X_B/3 > X_A$ on each battlefield secures the entire value of 1, whereas a single pre-commitment cannot secure the entire value. Note this assertion is different from Lemma \ref{lem:single_t_optimal}, which establishes that any multi-battlefield pre-commitment weakly under-performs a pre-commitment to a single battlefield on \emph{some} instance of battlefields. 

We conducted numerical studies to verify Lemma \ref{lem:single_t_optimal} and to compare the overall performance of single battlefield pre-commitments to the performance of pre-commitments to two battlefields in a three-battlefield setting with fixed total value $\phi=1$. The limit value is varied from 0 to 1. For each particular limit value, we sampled 500 valuations $\bs{v}$ uniformly from $\mcal{V}$. For each sample, we evaluated the best pre-commitment $(p_1,p_2)$ to the first two battlefields. In addition, we evaluated the best pre-commitment to a single battlefield of value $v_1+v_2$, i.e. on a corresponding set of valuations $\bs{v}'$ with two battlefields $v_1' = v_1+v_2$ and $v_2' = v_3$. On average, single battlefield pre-commitments significantly outperform pre-commitments to two battlefields (see Figure \ref{fig:GL_sims}).

%%%%%%%%%%%%%%%%%%%%%%%%%%%%%%%%%%%%%%%%%%%%%%%%%%%%%%%
\section{Pre-commitments with asymmetric valuations}
% !TEX root = main.tex

In this section, we consider pre-commitments in two-player General Lotto games where the players' valuations of the battlefields are asymmetric and relatively different. The main finding in this section is that pre-committing can guarantee a payoff that outperforms the second-highest equilibrium payoff that is available in nominal play. We begin with preliminaries on equilibrium characterizations of our setup.

\subsection{Equilibria of the nominal game $\text{GGL}(X_A,X_B,\alpha)$}

Recall the two-battlefield setup specified as $\text{GGL}(X_A,X_B,\alpha)$ (Figure \ref{fig:GGL_scenario}). Specifically, $v_{A,1} = \alpha$, $v_{A,2} = 1-\alpha$, $v_{B,1} = 1-\alpha$, and $v_{B,2} = \alpha$ for some $\alpha \in (0,\frac{1}{2}]$. To proceeed, we first establish the following equilibrium characterizations.
% % GGL equil. characterization
% \begin{fact}
%     Consider $\text{GGL}(X_A,X_B,\bs{v}_A, \bs{v}_B)$. Define the set $\mcal{B}_A(\sigma) := \{j\in[n]: \frac{v_{A,j}}{v_{B,j}} > \sigma\}$ for any $\sigma > 0$. Define
%     \begin{equation}
%         \begin{aligned}
%             F(\sigma) &:= \sigma^2\left(\sigma \sum_{j\in\mcal{B}_A(\sigma)} \frac{v_{B,j}^2}{v_{A,j}} - \frac{X_B}{X_A}\sum_{j\in\mcal{B}_A(\sigma)} v_{B,j} \right) \\
%             &\quad + \left(\sigma\sum_{j\notin\mcal{B}_A(\sigma)} v_{A,j}- \frac{X_B}{X_A}\sum_{j\notin\mcal{B}_A(\sigma)} \frac{v_{A,j}^2}{v_{B,j}} \right) 
%         \end{aligned}
%     \end{equation}
%     Then each equilibrium of $\text{GGL}(X_A,X_B,\bs{v}_A, \bs{v}_B)$ corresponds to a zero of $F(\sigma)$. For any zero $\sigma^*$, the expected payoffs are given by
%     \begin{equation}
%         \begin{aligned}
%             \pi_A(\sigma^*) &= \sum_{j\in\mcal{B}_A(\sigma^*)} \left( v_{A,j} - \frac{\sigma^* v_{B,j}}{2} \right) + \sum_{j\notin\mcal{B}_A(\sigma^*)} \frac{v_{A,j}^2}{2\sigma^* v_{B,j}} \\ \pi_B(\sigma^*) &= \sum_{j\notin\mcal{B}_A(\sigma^*)} \left( v_{B,j} - \frac{v_{A,j}}{2\sigma^*} \right) + \sum_{j\in\mcal{B}_A(\sigma^*)} \frac{v_{B,j}^2\sigma^*}{2 v_{A,j}}
%         \end{aligned}
%     \end{equation}
% \end{fact}
\begin{fact}[Using results from \cite{Kovenock_2020}]
	Each equilibrium of $\text{GGL}(X_A,X_B,\alpha)$ corresponds to a zero of the function $\sf(\sigma)$, given by
	\begin{equation}\label{eq:soln_func}
	\sf(\sigma) := 
	\begin{cases}
	\sigma^2\left[ \sigma(\frac{(1-\alpha)^2}{\alpha} + \frac{\alpha^2}{1-\alpha}) - \frac{X_A}{X_B} \right], &\!\!\!\!\!\sigma \in (0,\frac{\alpha}{1-\alpha}) \\
	\frac{\alpha^2}{1-\alpha}(\sigma^3 - \frac{X_A}{X_B}) + \alpha\sigma(1 - \frac{X_A}{X_B}\sigma), &\!\!\!\!\!\sigma \in [\frac{\alpha}{1-\alpha},\frac{1-\alpha}{\alpha}) \\
	\sigma - \frac{X_A}{X_B}(\frac{\alpha^2}{1-\alpha} + \frac{(1-\alpha)^2}{\alpha} ), &\!\!\!\!\!\sigma \geq \frac{1-\alpha}{\alpha}
	\end{cases}
	\end{equation}
	For any zero $\sigma^*$, the expected equilibrium payoffs are given by
	\begin{equation}\label{eq:pi_GGL_nom}
	\begin{aligned}
	\pi_A(\sigma^*) &= 
	\begin{cases}
	% \frac{X_A}{2X_B}, \quad &\sigma^* \in (0,\frac{\alpha}{1-\alpha}) \\
	\frac{\sigma^*}{2}(\frac{(1-\alpha)^2}{\alpha} + \frac{\alpha^2}{1-\alpha}), \quad &\sigma^* \in (0,\frac{\alpha}{1-\alpha}) \\ 
	1-\alpha-\frac{\alpha}{2\sigma^*}+\frac{\alpha^2\sigma^*}{2(1-\alpha)}, \quad &\sigma^* \in [\frac{\alpha}{1-\alpha},\frac{1-\alpha}{\alpha}) \\
	1 - \frac{1}{2\sigma^*}, \quad &\sigma^* \geq \frac{1-\alpha}{\alpha}
	\end{cases} \\
	\pi_B(\sigma^*) &=
	\begin{cases}
	1 - \frac{\sigma^*}{2}, &\sigma^* \in (0,\frac{\alpha}{1-\alpha}) \\
	1-\alpha-\frac{\alpha\sigma^*}{2}+\frac{\alpha^2}{2\sigma^*(1-\alpha)} , &\sigma^* \in [\frac{\alpha}{1-\alpha},\frac{1-\alpha}{\alpha}) \\
	\frac{1}{2\sigma^*}(\frac{(1-\alpha)^2}{\alpha} + \frac{\alpha^2}{1-\alpha}), &\sigma^* \geq \frac{1-\alpha}{\alpha}
	% \frac{X_B}{2X_A}, &\sigma^* \geq \frac{1-\alpha}{\alpha}
	\end{cases} \\
	\end{aligned}
	\end{equation} 
	Note that $\sigma^* = \frac{X_A}{X_B}(\frac{(1-\alpha)^2}{\alpha} + \frac{\alpha^2}{1-\alpha})^{-1} < \frac{\alpha}{1-\alpha}$ is the unique root of $S$ on the interval $(0,\frac{\alpha}{1-\alpha})$, which gives $\pi_A(\sigma^*) = \frac{X_A}{2X_B}$. Additionally, $\sigma^* = \frac{X_A}{X_B}(\frac{(1-\alpha)^2}{\alpha} + \frac{\alpha^2}{1-\alpha}) \geq \frac{1-\alpha}{\alpha}$ is the unique root of $S$ on the interval $[\frac{1-\alpha}{\alpha},\infty)$, which gives $\pi_B(\sigma^*) = \frac{X_B}{2X_A}$.
\end{fact}
The function $\pi_A(\sigma^*)$ is increasing in $\sigma^*$, while $\pi_B(\sigma^*)$ is decreasing in $\sigma^*$. We will refer to $\sf$ as the \emph{solution function}. It has been established that $\sf(\sigma)$ is continuous and has at least one zero \cite{Kovenock_2020}. In the result below, we show that $\sf$ can either have a unique zero or have three zeros, depending on the parameters. Such regions are depicted in Figure \ref{fig:GGL}.

\begin{lemma}\label{lem:multiplicity_regions}
	The solution function $\sf$ has three zeros if and only if $\frac{X_B}{X_A} \leq 1$ and $\sf(\sigma_-) > 0$, or $\frac{X_B}{X_A} > 1$ and $\sf(\sigma_+) < 0$, where $\sigma_\pm := \frac{1-\alpha}{3\alpha}\left[\frac{X_A}{X_B} \pm \sqrt{\left(\frac{X_A}{X_B}\right)^2 - \frac{3\alpha}{1-\alpha} } \right]$. Otherwise, $\sf$ has a unique zero.
\end{lemma}
The proof is given in the Appendix. The equilibrium strategies are given as follows.
\begin{fact}[Using results from \cite{Kovenock_2020}]\label{fact:GGL_strats}
	Let $\sigma^*$ be a zero of $\sf$, and let $\Omega(\sigma^*) = \{b\in\{1,2\} : \frac{v_{B,b}}{v_{A,b}} \geq \sigma^*\}$. The corresponding unique equilibrium strategies are given as follows. For all $b \in \Omega(\sigma^*)$, 
	\begin{equation}\label{eq:priority_bf}
	\begin{aligned}
	F_{A,b}(x) &= (1-\frac{v_{A,b}}{v_{B,b}}\sigma^*) + \frac{\lambda_B^*}{v_{B,b}}x, \quad  x \in [0,\frac{v_{A,j}}{\lambda_A^*}) \\
	F_{B,b}(x) &= \frac{\lambda_A^*}{v_{A,b}}x, \quad x \in [0,\frac{v_{A,j}}{\lambda_A^*}).
	\end{aligned}
	\end{equation}
	For all $b \notin \Omega(\sigma^*)$,
	\begin{equation}
	\begin{aligned}
	F_{A,b}(x) &= \frac{\lambda_B^*}{v_{B,b}}x, \quad x \in [0,\frac{v_{B,j}}{\lambda_B^*}) \\
	F_{B,b}(x) &= (1-\frac{v_{B,b}}{v_{A,b}\sigma^*}) + \frac{\lambda_A^*}{v_{A,b}}x, \quad x \in [0,\frac{v_{B,j}}{\lambda_B^*}) .
	\end{aligned}
	\end{equation}
	Here, $\lambda_A^* := \frac{1}{2X_B}\left(\sum_{b\in\Omega(\sigma^*)} v_{A,b} + \frac{1}{(\sigma^*)^2}\sum_{b\notin \Omega(\sigma^*)} \frac{v_{B,b}^2}{v_{A,b}} \right)$ and $\lambda_B^* := \sigma^*\lambda_A^*$.
\end{fact}
The set $\Omega(\sigma^*)$ can be interpreted as the set of battlefields that player $B$ places priority on in the equilibrium associated with $\sigma^*$, because it competes more aggressively on these battlefields than player $A$ does \eqref{eq:priority_bf}. Indeed, in a similar manner as the equilibrium strategies from GL games, player $A$ ``gives up" on each of these battlefields with a probability $1-\frac{v_{A,b}}{v_{B,b}}\sigma^*$. If $\sigma^* \in [\frac{\alpha}{1-\alpha},\frac{1-\alpha}{\alpha})$, for example, then $B$ allocates more resources in expectation to its more valuable $b=2$ than player $A$ does, and player $A$ allocates more resources to battlefield 1 in expectation. 

In the forthcoming analysis on pre-commitments in the GGL game, we will not need to make explicit use of the details of Fact \ref{fact:GGL_strats} since we are concerned only with the associated payoffs. Our results rely heavily on the characterization of the possibly non-unique equilibrium payoffs in the GGL game \eqref{eq:pi_GGL_nom}.

% %
% \begin{figure*}[t]
% 	\centering
% 	\hspace{-3mm}
% 	\begin{subfigure}{.32\textwidth}
% 		\centering
% 		\includegraphics[scale=.4]{figs/multiplicity_regions.eps}
% 		\caption{}\label{fig:GGL_regions}
% 	\end{subfigure}  
% 	\begin{subfigure}{.32\textwidth}
% 		\centering
% 		\includegraphics[scale=.28]{figs/uBt.eps}
% 		\caption{}\label{fig:uBt}
% 	\end{subfigure}
% 	\begin{subfigure}{.32\textwidth}
% 		\centering
% 		\includegraphics[scale=.28]{figs/uAt.eps}
% 		\caption{}\label{fig:uAt}
% 	\end{subfigure}
% 	\caption{\small (a) Parameters that admit three equilibria in $\text{GGL}(X_A,X_B,\alpha)$ is represented by the union of the blue and red regions (Lemma \ref{lem:multiplicity_regions}). The blue region indicates where there exist pre-commitments (to the first battlefield, worth $1-\alpha$) ensuring player $B$ a payoff better than the second-highest equilibrium payoff (Theorem \ref{thm:second_best}). In the red region, multiple equilibria exist but pre-commitments cannot do better than the second-highest equilibrium payoff.  (b,c) Subsequent payoffs for pre-commitments $p \in [0,X_B]$ to the first battlefield. Here, $X_A = 1$. $X_B = 0.8$, and $\alpha = 0.05$. The dotted horizontal lines indicate the three equilibrium payoffs in $\text{GGL}(X_A,X_B,\alpha)$. Player B's  has pre-commitments $p > 0.15$ \eqref{eq:B_indifference_pts_GGL} that ensure a payoff  exceeding the second-highest equilibrium payoff (Theorem \ref{thm:second_best}).}\label{fig:GGL}
% \end{figure*}
% %

\subsection{Guaranteed payoffs through pre-commitments}

Suppose player $B$ has the option to pre-commit $p\in[0,X_B]$ to one of the battlefields $b \in \{1,2\}$. If player $A$ decides to match the pre-commitment $p \leq X_A$ on $b$, its payoff is given by
\begin{equation}
\begin{aligned}
u_A^{\texttt{M}}(p) = v_{A,b} + (1-v_{A,b})\cdot L(X_A-p,X_B-p).
\end{aligned}
\end{equation}
That is, $A$ secures $b$ and both players use their remaining resources to compete on the other battlefield. We note that a player's equilibrium payoff is unique in a Generalized Lotto game that has a single battlefield, and is given by the function $L$ scaled by its valuation of that battlefield. If player $A$ decides to withdraw, its payoff is given by
\begin{equation}
\begin{aligned}
u_A^{\texttt{W}}(p) = (1-v_{A,b})\cdot L(X_A,X_B-p).
\end{aligned}
\end{equation}
Player $A$ will make the choice, match or withdraw, that optimizes its subsequent payoff:
\begin{equation}\label{eq:uA_GGL}
u_A(p) = 
\begin{cases}
\max\left\{u_A^{\texttt{M}}(p),u_A^{\texttt{W}}(p) \right\}, &\text{if } p \leq X_A \\
u_A^{\texttt{W}}(p), &\text{if } p > X_A
\end{cases}
\end{equation}
Let $\texttt{A}_b(p) \in \{\texttt{M},\texttt{W}\}$ denote player A's optimal response to the pre-commitment $p$ on $b$. Player B's subsequent payoff is thus given by 
\begin{equation}\label{eq:uB_GGL}
u_B(p) = 
\begin{cases}
(1 - v_{B,b}) L(X_B-p,X_A-p), &\!\!\!\text{if } \texttt{A}_b(p) = \texttt{M} \\
v_{B,b} + (1 - v_{B,b}) L(X_B-p,X_A), &\!\!\!\text{if } \texttt{A}_b(p) = \texttt{W}
\end{cases}
\end{equation}

\begin{figure}[t]
	\centering
	% 	\hspace{-3mm}
	\includegraphics[scale=.4]{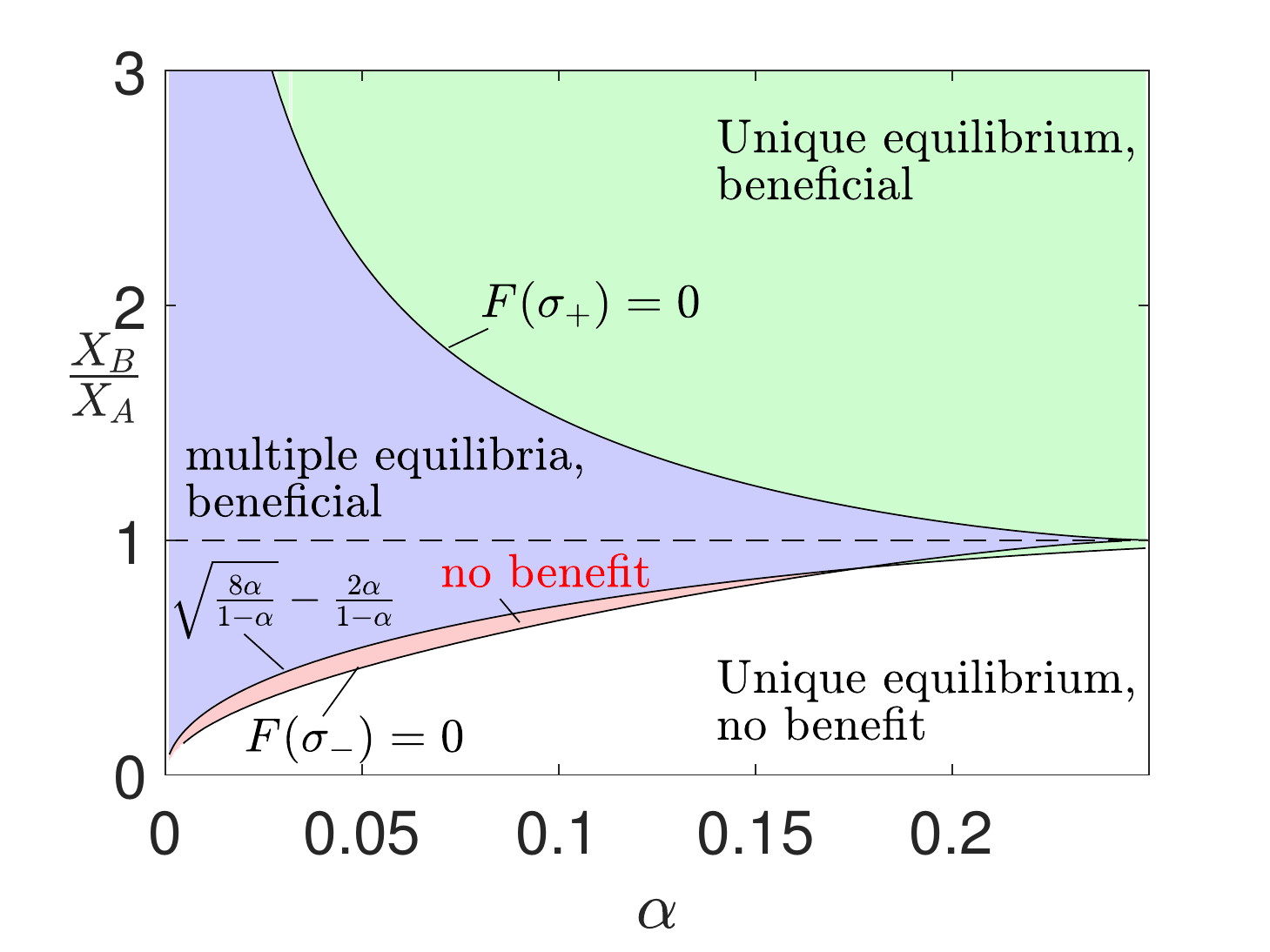}
	\caption{\small Parameters that admit three equilibria in $\text{GGL}(X_A,X_B,\alpha)$ is represented by the union of the blue and red regions (Lemma \ref{lem:multiplicity_regions}). The blue region indicates where there exist pre-commitments (to the first battlefield, worth $1-\alpha$) ensuring player $B$ a payoff better than the second-highest equilibrium payoff (Theorem \ref{thm:second_best}). In the red region, multiple equilibria exist but pre-commitments cannot do better than the second-highest equilibrium payoff.}\label{fig:GGL}
\end{figure}

The subsequent payoffs to both players are uniquely determined when player $B$ elects to pre-commit (\eqref{eq:uA_GGL} and \eqref{eq:uB_GGL}). In scenarios where the nominal game $\text{GGL}(X_A,X_B,\alpha)$ admits multiple payoff-distinct equilibria, one can thus view pre-committing as a way to guarantee a certain level of payoff. We will show there are pre-commitments that ensure a payoff that exceed the second-highest equilibrium payoff from the nominal, simultaneous move game $\text{GGL}(X_A,X_B,\alpha)$, regardless of whether player $B$ is the weaker or stronger player.

We can now compare the quality of the three equilibria in the nominal GGL game to the payoff player $B$ can ensure through a pre-commitment \eqref{eq:uB_GGL}. In particular, we identify parameters for which there are pre-commitments that ensure a better payoff than the second-highest equilibrium payoff. Let us denote $Z$ as the set of zeros of $F$, which by Lemma \ref{lem:multiplicity_regions}, is either a singleton or contains three zeros. In cases where it has three zeros, let us denote $\sigma_1^*,\sigma_2^*$, and $\sigma_3^*$ as the zeros of $F$ corresponding to the best, middle, and worst equilibrium payoff available to player B, respectively. In particular, these satisfy $\sigma_1^*<\sigma_2^*<\sigma_3^*$ and $\pi_B(\sigma_1^*) > \pi_B(\sigma_2^*) > \pi_B(\sigma_3^*)$. 

First, we establish that when player $B$ is weaker, it never has an incentive to pre-commit to the second (non-priority) battlefield, which has value $v_{B,2} = \alpha < \frac{1}{2}$. The result below shows that any pre-commitment to battlefield 2 yields a subsequent payoff worse than the worst-case equilibrium payoff from the nominal GGL game.

\begin{lemma}\label{lem:PC_BF2}
	Suppose $\frac{X_B}{X_A} \leq 1$. If player $B$ pre-commits $p\in[0,X_B]$ to battlefield 2, which has value $v_{B,2} = \alpha < \frac{1}{2}$ to player $B$, then $u_B(p) < \min_{\sigma^* \in Z} \pi_B(\sigma^*)$. 
\end{lemma}
\begin{proof}
	We show that player $A$ matches any pre-commitment $p\in[0,X_B]$. The condition $u_A^\texttt{M}(p) > u_A^\texttt{W}(p)$ for all $p\in[0,X_B]$ is equivalent to
	\begin{equation}
	\frac{\alpha}{2X_A}p^2 - (1 - \alpha(1-\frac{X_B}{2X_A}))p + \alpha(X_B-2X_A) + X_A > 0
	\end{equation}
	for all $p\in[0,X_B]$. The condition holds by observing this function is strictly decreasing for $p\in[0,X_B]$, and its values at the endpoints are positive. Therefore, player B's subsequent payoff is given by $u_B(p) = (1-\alpha)\frac{X_B-p}{2(X_A-p)}$. Note this is strictly decreasing in $p$, so $\max_{p\in[0,X_B]} u_B(p) = u_B(0) = (1-\alpha)\frac{X_B}{2X_A}$.
	
	When $\frac{X_B}{X_A} \leq 1$, player $B$'s worst-case payoff satisfies $\min_{\sigma^*\in Z}\pi_B(\sigma^*) \geq \frac{X_B}{2X_A} > u_B(0)$. Hence, any pre-commitment leads to a payoff lower than the worst-case equilibrium of the nominal GGL game.
	
	%determined is follows. When $\frac{X_B}{X_A} < (1-\alpha)(1+(\frac{\alpha}{1-\alpha})^3)$, $\sigma_3^* \in [\frac{1-\alpha}{\alpha},\infty)$, and hence $\pi_B(\sigma_3^*) = \frac{X_B}{2X_A}> u_B(0)$. When $\frac{X_B}{X_A} \geq (1-\alpha)(1+(\frac{\alpha}{1-\alpha})^3)$, $\sigma_3^* \in [\frac{\alpha}{1-\alpha},\frac{1-\alpha}{\alpha})$, and the payoffs satisfy $\pi_B(\sigma_3^*) > \pi_B(\frac{1-\alpha}{\alpha}) = \frac{X_B}{2X_A} > u_B(0)$. 
\end{proof}

% Figures \ref{fig:uBt} and \ref{fig:uAt} illustrate the subsequent payoffs $u_A(p)$, $u_B(p)$ in a particular instance for all pre-commitments $p \in [0,X_B]$ to battlefield 1.

Note that this result holds in all cases, i.e. whether there is a unique equilibrium (conditions of Lemma \ref{lem:multiplicity_regions} do not hold) or there are multiple equilibria. The following result demonstrates that when there are multiple equilibria, and under mild conditions of the parameters, pre-commitments can offer better a subsequent payoff than the second-highest equilibrium payoff available in $\text{GGL}(X_A,X_B,\alpha)$.

\begin{theorem}\label{thm:second_best}
	Suppose there are three equilibria in the nominal GGL game, i.e. conditions of Lemma \ref{lem:multiplicity_regions} are met. Then there are pre-commitments for player $B$ that ensure its payoff \eqref{eq:uB_GGL} exceeds the second-highest equilibrium payoff $\pi_B(\sigma_2^*)$ \eqref{eq:pi_GGL_nom} if $\sqrt{\frac{8\alpha}{1-\alpha}} - \frac{2\alpha}{1-\alpha} \leq \frac{X_B}{X_A} \leq 1$, or $\frac{X_B}{X_A} > 1$.
\end{theorem}
\begin{proof}
	To prove the result, we first establish that $\pi_B(\sigma_2^*) < 1-\alpha$. Then, under the given assumptions, we show there are pre-commitments to battlefield 1 that give a payoff $u_B(t) > 1-\alpha$ because they cause player $A$ to withdraw.
	
	In the nominal GGL game, the middle equilibrium payoff is $\pi_B(\sigma_2^*) = 1-\alpha + \frac{\alpha}{2}(\frac{\alpha }{\sigma_2^*(1-\alpha)} - \sigma_2^*) < 1-\alpha$. The inequality follows by establishing $\sigma_2^* > \sqrt{\frac{\alpha}{1-\alpha}}$: since $\sigma_- < \sigma_2^*$, it will suffice to show that $\sigma_- > \sqrt{\frac{\alpha}{1-\alpha}}$. This is equivalent to the condition $\frac{X_A}{X_B} < \frac{3(1+(\frac{\alpha}{1-\alpha})^2)}{2\sqrt{\frac{\alpha}{1-\alpha}}}$. We observe that $(\sqrt{\frac{8\alpha}{1-\alpha}} - \frac{2\alpha}{1-\alpha})^{-1} < \frac{3(1+(\frac{\alpha}{1-\alpha})^2)}{2\sqrt{\frac{\alpha}{1-\alpha}}}$ for all $\alpha \in [0,\frac{1}{2}]$, since this reduces to $1 < (\sqrt{2}+\sqrt{\frac{\alpha}{1-\alpha}})(1 + (\frac{\alpha}{1-\alpha})^2 )$.
	
	Now, suppose $\frac{X_B}{X_A} \leq 1$ and player $B$ pre-commits to battlefield 1 (Lemma \ref{lem:PC_BF2}). Player $A$ is indifferent between $\texttt{M}$ and $\texttt{W}$ (eq. \eqref{eq:uA_GGL}) for the pre-commitment
	\small
	\begin{equation}\label{eq:B_indifference_pts_GGL}
	{\small
		p_\pm = \frac{X_A}{2}\left[\left(\frac{X_B}{X_A}+\frac{2\alpha}{1-\alpha}\right) \pm \sqrt{\left(\frac{X_B}{X_A}+\frac{2\alpha}{1-\alpha}\right)^2 - \frac{8\alpha}{1-\alpha}} \right].}
	\end{equation}
	\normalsize
	Such values exist only if $\frac{X_B}{X_A} \geq \sqrt{\frac{8\alpha}{1-\alpha}} - \frac{2\alpha}{1-\alpha}$, where it follows that $p_\pm \in [0,X_B]$. Here, $\texttt{A}(p) = \texttt{M}$ for $p \in [0,p_-) \cup [p_+,X_B]$, and $\texttt{A}(p) = \texttt{M}$ for $p \in [p_-,p_+)$. For the pre-commmitment $p_-$, player B's subsequent payoff is $u_B(p_-) = (1-\alpha) + \alpha\frac{X_B-p_-}{2X_A} > 1-\alpha$.  
	
	Suppose $\frac{X_B}{X_A} > 1$, and player $B$ pre-commits to battlefield 1. If $\frac{X_B}{X_A} \geq  \frac{1+\alpha}{1-\alpha}$, player $A$ is indifferent between \texttt{M} and \texttt{W}  for the pre-commitment $p = \frac{2\alpha}{1+\alpha}X_B \leq X_B - X_A$. If $\frac{X_B}{X_A} <  \frac{1+\alpha}{1-\alpha}$, player $A$ is indifferent between \texttt{M} and \texttt{W} for the pre-commitment $p = X_B - \frac{X_A}{2(1-\alpha)}\left[ (1-3\alpha) + \sqrt{(1-3\alpha)^2 + 4(1-\alpha)^2(\frac{X_B}{X_A} - 1)}\right] > X_B - X_A$. In either case, the pre-commitment $p$ induces player $A$ to withdraw, and the subsequent payoff is $u_B(p)=1-\alpha + \alpha\frac{X_B-p}{2X_A} > 1-\alpha$.
\end{proof}

It is interesting that player $B$ can have pre-commitments that ensure a payoff greater than $\pi_B(\sigma_2^*)$ whether it is the weaker or stronger player. Additionally, when player $B$ is stronger, the amount it needs to pre-commit to ensure a payoff greater than $\pi_B(\sigma_2^*)$ is less than player A's budget $X_A$. These aspects contrast with the result of Theorem \ref{thm:GL_result}: when player $B$ is weaker in the standard GL game, there are no pre-commitments that outperform the nominal equilibrium payoff. When player $B$ is stronger, a pre-commitment $p>X_A$ is required, i.e. to force player $A$ to withdraw. 

When $\text{GGL}(X_A,X_B,\alpha)$ admits a unique equilibrium, we identify below a region where there are beneficial pre-commitments, even when player $B$ is the weaker player. 
\begin{theorem}\label{thm:GGL_weaker_unique_benefit}
	Suppose that $\sqrt{\frac{8\alpha}{1-\alpha}} - \frac{2\alpha}{1-\alpha} \leq \frac{X_B}{X_A} < \min\left\{1,\sqrt{\frac{1-\alpha}{3\alpha}}\right\}$ and there is a unique equilibrium of the nominal GGL game. Then there are precommitments for player $B$ that ensures its payoff exceeds the (unique) nominal equilibrium payoff.
\end{theorem}
\begin{proof}
	The condition $\frac{X_B}{X_A} < \min\left\{1,\sqrt{\frac{1-\alpha}{3\alpha}}\right\}$ means both critical points $\sigma\pm$ are in the second interval $[\frac{\alpha}{1-\alpha},\frac{1-\alpha}{\alpha})$. The unique root $\sigma^*$ of $\sf$ also lies in the second interval, and satisfies $\sigma- < \sigma^*$. Thus, $\pi_A(\sigma^*) < 1-\alpha$, by using the same reasoning in the proof of Theorem \ref{thm:second_best}. The condition $\frac{X_B}{X_A} \geq \sqrt{\frac{8\alpha}{1-\alpha}} - \frac{2\alpha}{1-\alpha} $ ensures there are pre-commitments that force player $A$ to withdraw. Thus, player $B$ can get at least a payoff of $1-\alpha$ through a pre-commitment.
\end{proof}

%%%%%%%%%%%%%%%%%%%%%%%%%%%%%%%%%%%%%%%%%%%%%%%%%%%%%%%
\section{Conclusion}
This paper investigated competitive scenarios where publicly announcing strategic intentions can have advantages. Conventional wisdom would suggest that this should not be the case, since known strategies can be exploited by an opponent. We studied a formulation using General Lotto games that allows one of the players to pre-commit resources before engaging in competition. A main finding was that the weaker-resource player never has an incentive to do so, but a stronger-resource player has incentives under certain conditions. We then studied a setting where battlefields are not valued identically by the players. Without pre-commitments, this interaction can yield multiple payoff-distinct equilibria. We found that a weaker-resource player does have incentives to pre-commit resources. However, these incentives arise under a different context, where pre-committing can offer better performance than the second-highest equilibrium payoff in the game without pre-commitments.

\bibliographystyle{IEEEtran}
\bibliography{sources}

%%%%%%%%%%%%%%%%%%%%%%%%%%%%%%%%%%%%%%%%%%%%%%%%%
%=== ================APPENDIX ============================================
%%%%%%%%%%%%%%%%%%%%%%%%%%%%%%%%%%%%%%%%%%%%%%%%%
\appendices
% !TEX root = main.tex

% \section{}\label{sec:GL_appendix}
\begin{proof}[Proof of Theorem \ref{thm:GL_result}]
	Suppose player $B$ is weaker, i.e. $X_B < X_A$. Suppose player B pre-commits $p \in [0,X_B]$ resources to battlefield $b \in \mcal{B}$ that has any value ${v_b} \in [0,\phi]$. The claim is that player $B$ cannot attain a subsequent payoff $u_B(p) > \pi_B^{\texttt{nom}}$ from any pre-commitment. To see this, consider the case that $A$ matches the pre-commitment. Its payoff is
	\begin{equation}
	\begin{aligned}
	u_{A,\{b\}}(p) &= v_b + (\phi-v_b)(1 - \frac{X_B-p}{2(X_A-p)}) \\
	&= \phi - (\phi-v_b)\frac{X_B-p}{2(X_A-p)}
	\end{aligned}
	\end{equation}
	This payoff always exceeds its nominal equilibrium payoff $\pi_A^\texttt{nom} = \phi(1 - \frac{X_B}{2X_A}$), since the function $\frac{X_B-p}{X_A-p}$ is decreasing on $p \in [0,X_A)$. Note that $u_A(p) = \pi_A^\texttt{nom}$ if and only if $(p,v) = (0,0)$, i.e. there is no pre-commitment. Player $A$'s subsequent payoff thus satisfies $u_A(p) = \max\{u_{A,\{b\}}(p),u_{A,\varnothing}(p)\} > \pi_A^\texttt{nom}$. Therefore, $u_B(p) = \phi - u_A(p) \leq \pi_B^\texttt{nom}$, with equality if and only if $(p,v) = (0,0)$.

	% % % % % % % % % % % % % % % % % % % % % % % % % % % % % % % % % % % % % % % %
	% % % % PROOF OF 1v1 STRONGER % % % %
	% % % % % % % % % % % % % % % % % % % % % % % % % % % % % % % % % % % % % % % %
	
	Now, we illustrate the proof for $X_A < X_B \leq 2X_A$. Let us denote $\gamma := \frac{X_B}{X_A}$ as the budget ratio, so that $\gamma \in (1,2]$. Player $A$'s payoff from matching the pre-commitment is given by
	\begin{equation}
	u_{A,\{b\}}(p) = v_b + (\phi-v_b) \frac{X_A-p}{2(X_B-p)}
	\end{equation}
	Player $A$'s payoff from withdrawing against the pre-commitment is given by
	\begin{equation}
	u_{A,\varnothing}(p) =
	\begin{cases}
	(\phi-v_b)\frac{X_A}{2(X_B-p)}, &\text{if } p \leq X_B - X_A \\
	(\phi-v_b)(1 - \frac{X_B-p}{2X_A}), &\text{if } p > X_B - X_A
	\end{cases}
	\end{equation}
	If $v_b > \frac{3-\gamma}{5-\gamma}\phi$, then Player $B$'s subsequent payoff is then given by
	\begin{equation}
	u_B(p) = 
	\begin{cases}
	\phi - u_{A,\{b\}}(p), &\text{if } p \in [0,X_A] \\
	\phi - (\phi-v_b)(1 - \frac{X_B-p}{2X_A}), &\text{if } p \in (X_A,X_B]
	\end{cases}
	\end{equation}
	The above expression is increasing on $p\in[0,X_A)$ ($A$ matches) and decreasing on $p \in [X_A,X_B]$ ($A$ withdraws). We observe that player $B$ cannot exceed its nominal payoff for any pre-commitment $p\in[0,X_A)$. Indeed, the condition $\phi - u_{A,\{b\}}(X_A) < \pi_B^\texttt{nom}$ is equivalent to $v_b>\frac{\phi}{2\gamma}$. It suffices to show that $\frac{\phi}{2\gamma} < \frac{3-\gamma}{5-\gamma}\phi$. This reduces to $\gamma^2 - \frac{7}{2}\gamma + \frac{5}{2} < 0$, which is satisfied for all $\gamma \in (1,\frac{5}{2})$.    
	
	The right-limit value $\lim_{p\searrow X_A} u_B(p) = \phi - \frac{1}{2}(\phi-v_b)(3-\gamma)$ exceeds the nominal equilibrium payoff if and only if $v_b > (1 - \frac{1}{\gamma(3-\gamma)})\phi$. Therefore, there exist beneficial pre-commitments $p>X_A$ for player $B$ if and only if the limit value satisfies $\bar{v} > (1 - \frac{1}{\gamma(3-\gamma)})\phi$.
	
	For the remainder of the proof, we show there are no beneficial pre-commitments for all other parameters. Let us consider the regime $v_b \in [\frac{\gamma-1}{\gamma+1}\phi,\frac{3-\gamma}{5-\gamma}\phi]$ (need $\gamma < 2$). Player $A$ becomes indifferent between matching and withdrawing under the pre-commitment $p^* = X_B - \frac{X_A}{2(\phi-v_b)}\left[(\phi-3v_b) + \sqrt{(\phi-3v_b)^2 + 4(\phi-v_b)^2(\gamma-1)} \right] \leq X_A$. Player $B$'s subsequent payoff is then given by
	\begin{equation}\label{eq:middle_region}
	u_B(p) = 
	\begin{cases}
	\phi - u_{A,\{b\}}(p), &\text{if } p \in [0,p^*) \\
	\phi - (\phi-v_b)(1 - \frac{X_B-p}{2X_A}), &\text{if } p \in [p^*,X_B]
	\end{cases}
	\end{equation}
	This is a continuous function in $p$, which is increasing on $p \in [0,p^*)$ ($A$ matches) and decreasing on $p \in [p^*,X_B]$ ($A$ withdraws). It takes its maximal value at $p^*$. To prove $u_B(p^*) < \pi_B^{\texttt{nom}}$, it suffices to show that player $B$ has an incentive to pre-commit only if player $A$ matches a pre-commitment $p > p*$, which is impossible due to \eqref{eq:middle_region}. Indeed, $\phi - u_A^\texttt{M}(p) > \pi_B^{\texttt{nom}}$ reduces to $p > \frac{2X_Bv_b(1-\frac{1}{2\gamma})}{2\phi(1-\frac{1}{2\gamma}) - (\phi-v_b)}$. We will show the RHS is greater than $t^*$. Such a condition is equivalent to
	\begin{equation}
	\begin{aligned}
	(\gamma+1)v_b^2 &- \phi(2(\gamma-1)+\gamma^{-1})v_b \\
	&+\phi^2(\gamma-3+3\gamma^{-1}- \gamma^{-2}) \leq 0
	\end{aligned}
	\end{equation}
	Denote the above quadratic function as $G(v_b)$. We claim the above inequality holds for all $v_b \in [\frac{\gamma-1}{\gamma+1}\phi,\frac{3-\gamma}{5-\gamma}\phi]$. Since $G$ is convex, it suffices to show it is non-positive at the endpoints. Indeed, one can show that $G(\frac{\gamma-1}{\gamma+1}\phi) < 0$ and $G(\frac{3-\gamma}{5-\gamma}\phi) < 0$ for all $\gamma \in (1,2)$.
	% % Algebraic steps
	% \begin{equation}
	%     \begin{aligned}
	%         G(\frac{\gamma-1}{\gamma+1}\phi) &= \frac{\gamma-1}{\gamma}\phi^2\left[(\gamma-1)^2(\frac{1}{\gamma}-\frac{1}{\gamma+1}) - \frac{1}{\gamma+1} \right] \\
	%         &< \frac{\gamma-1}{\gamma}\phi^2\left[\frac{1}{\gamma}-\frac{2}{\gamma+1}\right] < 0.
	%     \end{aligned}
	% \end{equation}
	% The first inequality is due to $\gamma-1 < 1$, and the second is due to $\gamma > 1$. The value at the right endpoint is
	% \begin{equation}
	%     \begin{aligned}
	%         G(\frac{3-\gamma}{5-\gamma}\phi) &= \phi^2\left[\left(\frac{3-\gamma}{5-\gamma}\right)^2\!\!\!\!(\gamma+1) - \frac{3-\gamma}{5-\gamma}(2(\gamma-1) + \frac{1}{\gamma}) \right. \\
	%         &\quad\quad\quad + \left.\frac{1}{\gamma^2}(\gamma-1)^3  \right] < 0
	%     \end{aligned}
	% \end{equation}
	% for all $\gamma \in (1,2)$.

	Now, consider the (remaining) regime $v_b < \frac{\gamma-1}{\gamma+1}$. The pre-commitment $p^* = \frac{2v_b}{\phi+v_b}X_B \leq X_B-X_A$ renders player A to be indifferent between matching and withdrawing. Player $B$'s subsequent payoff is given by
	\begin{equation}
	u_B(p) = 
	\begin{cases}
	\phi - u_{A,\{b\}}(p), &\text{if } p \in [0,p^*) \\
	\phi - (\phi-v_b)\frac{X_A}{2(X_B-p)}, &\text{if } p \in [p^*,X_B-X_A) \\
	\phi - (\phi-v_b)(1 - \frac{X_B-p}{2X_A}), &\text{if } p \in [X_B-X_A,X_B]
	\end{cases}
	\end{equation}
	This is increasing on $p \in [0,p^*)$ ($A$ matches) and decreasing on $p \in [p^*,X_B]$ ($A$ withdraws). Thus, $u_B(p)$ is maximized at $p^*$, which yields
	\begin{equation}
	u_B(p^*) = \phi - \frac{\phi+v_b}{2\gamma} < \phi - \frac{\phi}{2\gamma} = \pi_B^\texttt{nom}
	\end{equation}
	Hence, there is no incentive for player $B$ to pre-commit any resources in this regime.
	
	% Now, suppose $\frac{3v_b-\phi}{\phi-v_b} \leq \gamma \leq \frac{\phi+v_b}{\phi-v_b}$. 
	The case when $\gamma \geq 2$ follows similar arguments. Here, there will only be two regimes: $v_b \leq \frac{\phi}{2\gamma-1}$, in which $p^*=\frac{2v_b}{\phi+v_b} \leq X_A$ is the pre-commitment that makes $A$ indifferent between matching and withdrawing, and $v_b > \frac{\phi}{2\gamma-1}$, in which $A$ matches any $p \in [0,X_A]$ and withdraws against any $p \in (X_A,X_B]$. We omit technical details for brevity, as they follow similar approaches to the regime $\gamma \in (1,2]$.
\end{proof}

% \section{}\label{sec:GGL_appendix}

\begin{proof}[Proof of Lemma \ref{lem:multiplicity_regions}]
	We give the proof for $\frac{X_B}{X_A} \leq 1$, since the case $\frac{X_B}{X_A} > 1$ follows similar arguments. 
	
	There are no zeros of $\sf$ in the first interval $\gamma \in [0,\frac{\alpha}{1-\alpha})$ since $\sf(\gamma) < 0$ for all  $\gamma \in (0,\frac{\alpha}{1-\alpha})$. In the second interval $\sigma \in [\frac{\alpha}{1-\alpha},\frac{1-\alpha}{\alpha})$, the solution function is negative at the left endpoint: $\sf(\frac{\alpha}{1-\alpha}) < 0$ follows from $\frac{X_A}{X_B} \geq 1$. The critical points are located at the two positive values $\sigma_{\pm}$. The first critical point, $\sigma_-$, is a local maximum and the second critical point, $\sigma_+$, is a local minimum. These points are defined as real numbers if and only if $\frac{X_A}{X_B} > \sqrt{\frac{3\alpha}{1-\alpha}}$. Now, suppose $\sf(\sigma_-) > 0$. We note that this assumption automatically implies that $\sigma_- \in (\frac{\alpha}{1-\alpha}, \frac{\alpha}{1-\alpha})$, i.e. it is in the second interval. We have $\sigma_- > \frac{\alpha}{1-\alpha}$ since $\sf(\sigma) < 0$ on the first interval. We have $\sigma_- < \frac{1-\alpha}{\alpha}$ because $\frac{X_A}{X_B} - \sqrt{\left(\frac{X_A}{X_B}\right)^2 - \frac{3\alpha}{1-\alpha} } < 3$ if $\frac{X_A}{X_B} < 3$, and $\frac{X_A}{X_B} > \frac{1}{2}\left(3+\frac{\alpha}{1-\alpha} \right)$ if $\frac{X_A}{X_B}> 3$.
	% The derivative is
	% \begin{equation}
	%     \frac{\partial \sf}{\partial \sigma}(\sigma) = 3\frac{\alpha^2}{1-\alpha}\sigma^2 - 2\alpha\frac{X_A}{X_B}\sigma + \alpha.
	% \end{equation}
	
	We know there is at least one zero in the interval $[\frac{\alpha}{1-\alpha},\frac{1-\alpha}{\alpha})$ because $\sf(\frac{\alpha}{1-\alpha}) < 0$. It will then suffice to show that there is a point $\hat\sigma \in (\sigma_-,\frac{1-\alpha}{\alpha}]$ such that $\sf(\hat\sigma) < 0$. This establishes there are at least two roots in the interval $[\frac{\alpha}{1-\alpha},\frac{1-\alpha}{\alpha})$. A third root $\sigma_3^* > \sigma_+$ is then guaranteed to exist because $\sf$ is strictly increasing (above a fixed rate) on $(\min\{\sigma_+,\frac{1-\alpha}{\alpha}\},\infty)$. 
	
	Suppose the solution function is negative at the right endpoint, i.e. $\sf(\frac{1-\alpha}{\alpha}) < 0$. This holds if and only if $\frac{X_A}{X_B} > \frac{(1-\alpha)^2}{\alpha^3 + (1-\alpha)^3}$.    In this regime, we can take $\hat{\sigma} = \frac{1-\alpha}{\alpha}$ and we are done. 
	
	Now, suppose $\sf(\frac{1-\alpha}{\alpha}) \geq 0$. Consider the point
	\begin{equation}
	\hat\sigma = \frac{1-\alpha}{3\alpha}\frac{X_A}{X_B}  > \sigma_-.
	\end{equation}
	We have $\hat\sigma < \frac{1-\alpha}{\alpha}$ if and only if $\frac{X_A}{X_B} < 3$, which is satisfied: $\frac{X_A}{X_B} \leq \frac{(1-\alpha)^2}{\alpha^3 + (1-\alpha)^3} < 3$. At this point, the solution function satisfies
	\begin{equation}
	\sf(\hat\sigma) \propto \frac{1}{3} + \frac{3\alpha}{1-\alpha}\left(1 - 3\left(\frac{\alpha}{1-\alpha} \right)^2 \right) - \left(\frac{X_A}{X_B}\right)^2.
	\end{equation}
	This is negative whenever $\frac{X_A}{X_B} > \sqrt{\frac{1}{3} + \frac{3\alpha}{1-\alpha}(1 - 3(\frac{\alpha}{1-\alpha} )^2 )}$. The expression inside the root takes its maximum value of 1 at $\alpha = \frac{1}{4}$. Hence, under the assumption that $\frac{X_A}{X_B}\geq 1$, the above condition is always met. Therefore, $\sf(\hat\sigma) < 0$.
\end{proof}

%----- Biography ----------------
\begin{IEEEbiography}[{\includegraphics[width=1in,height=1.25in,clip,keepaspectratio]{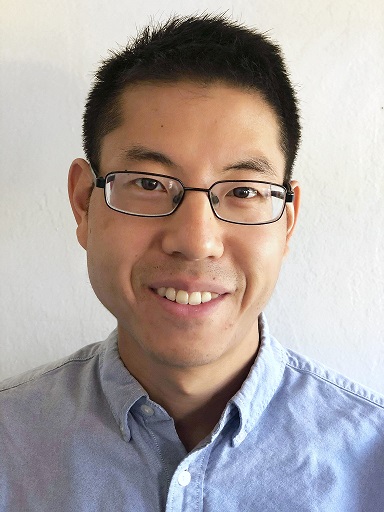}}] {Keith Paarporn}
	received the B.S. in Electrical Engineering from the University of Maryland, College Park in 2013, the M.S. in Electrical and Computer Engineering from the Georgia Institute of Technology in 2016, and the Ph.D. in Electrical and Computer Engineering from the Georgia Institute of Technology in 2018. He is currently a postdoctoral scholar in the Electrical and Computer Engineering Department at the University of California, Santa Barbara. His research interests include game theory, networked multi-agent systems, and control.
\end{IEEEbiography}

\begin{IEEEbiography}[{\includegraphics[width=1in,height=1.25in,clip,keepaspectratio]{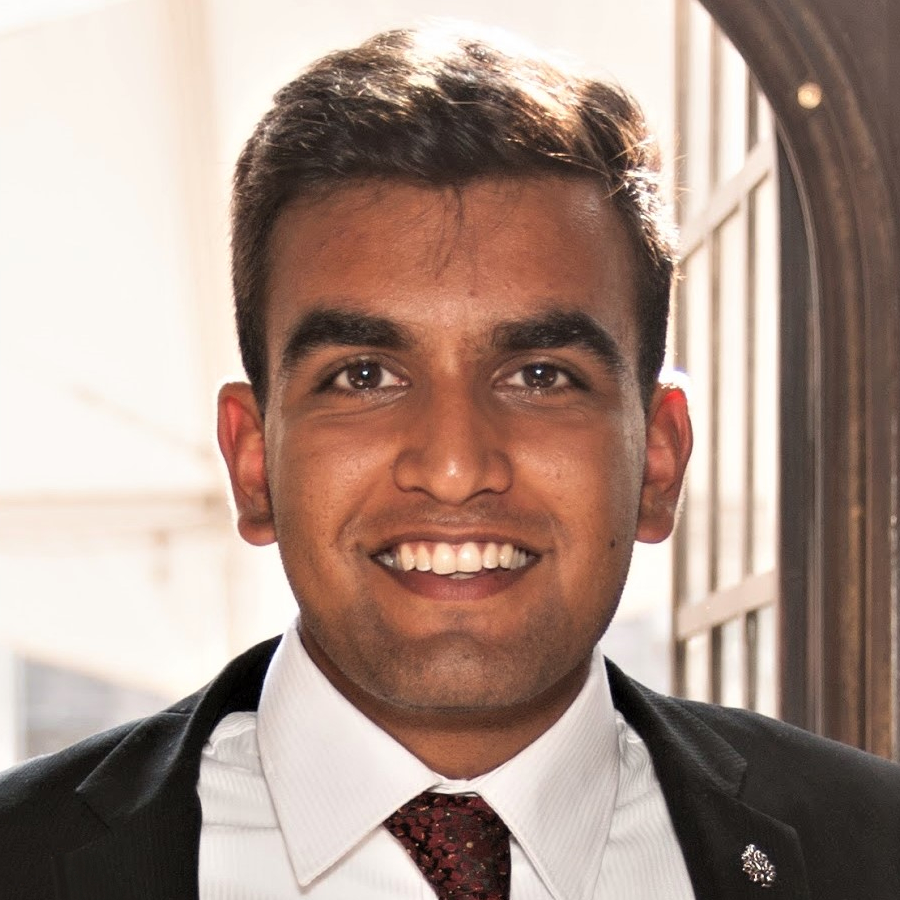}}] {Rahul Chandan}
	received the BASc in Engineering Science with a minor in Robotics and Mechatronics at the University of Toronto in 2017. He is currently pursuing a Ph.D in Electrical and Computer Engineering at the University of California, Santa Barbara. His research interests include economics and computation, game theory, and optimization.
\end{IEEEbiography}

\begin{IEEEbiography}[{\includegraphics[width=1in,height=1.25in,clip,keepaspectratio]{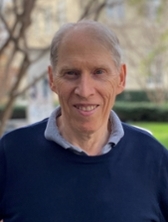}}] {Dan Kovenock}
	is Professor of Economics in the Economic Science Institute at Chapman University. He
	has also held professorships at the University of Iowa and Purdue University. He previously served as an
	Editor of the International Journal of Industrial Organization and on the editorial boards of the European
	Economic Review, Social Choice and Welfare, and the Strategic Management Journal. He is currently a
	Co-editor of Economic Theory and the Economic Theory Bulletin, and serves on the editorial boards of
	Games and Economic Behavior and the Journal of Public Economic Theory. He has more than 70 papers
	in international publications, including outlets such as the American Economic Review, the Review of
	Economic Studies, Management Science, Naval Research Logistics, Games and Economic Behavior, and
	Experimental Economics.
\end{IEEEbiography}

\begin{IEEEbiography}[{\includegraphics[width=1in,height=1.25in,clip,keepaspectratio]{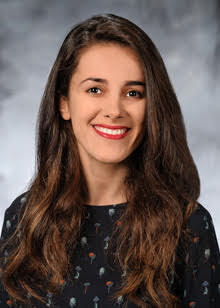}}] {Mahnoosh Alizadeh}
	received the B.Sc. degree in
	electrical engineering from the Sharif University
	of Technology in 2009, and the M.Sc. and Ph.D.
	degrees in electrical and computer engineering from
	the University of California at Davis in 2013 and
	2014, respectively. She is an Assistant Professor
	of electrical and computer engineering with the
	University of California at Santa Barbara. From
	2014 to 2016, she was a Post-Doctoral Scholar
	with Stanford University. Her research interests are
	focused on designing scalable control and market
	mechanisms for enabling sustainability and resiliency in societal infrastructures, with a particular focus on demand response and electric transportation
	systems. She was a recipient of the NSF CAREER Award.
\end{IEEEbiography}

\begin{IEEEbiography}[{\includegraphics[width=1in,height=1.25in,clip,keepaspectratio]{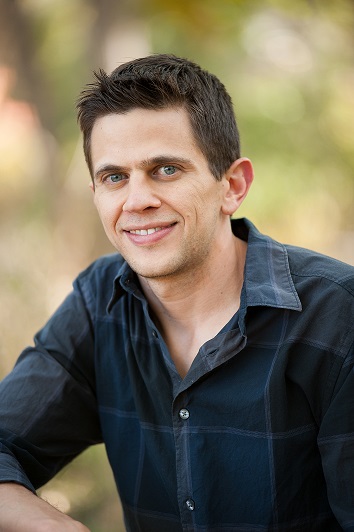}}] {Jason Marden}
	is an Associate Professor in the Department of Electrical and Computer Engineering at
	the University of California, Santa Barbara. Jason received a BS in Mechanical Engineering in 2001 from
	UCLA, and a PhD in Mechanical Engineering in
	2007, also from UCLA, under the supervision of Jeff
	S. Shamma, where he was awarded the Outstanding
	Graduating PhD Student in Mechanical Engineering.
	After graduating from UCLA, he served as a junior
	fellow in the Social and Information Sciences Laboratory at the California Institute of Technology until
	2010 when he joined the University of Colorado. Jason is a recipient of the NSF Career Award (2014), the ONR Young Investigator Award (2015), the AFOSR Young Investigator Award (2012), the American Automatic Control Council Donald P. Eckman Award (2012), and the SIAG/CST Best SICON Paper Prize (2015). Jason’s research interests focus on game theoretic methods for the control of distributed multiagent systems.
\end{IEEEbiography}

\end{document}